\documentclass[12pt]{amsart}
\usepackage{}

\usepackage{tikz}
\usetikzlibrary{arrows}

\usepackage{array}
\usepackage{amsmath}
\usepackage{amsfonts}
\usepackage{amssymb}
\usepackage{mathrsfs}

\usepackage[all,cmtip]{xy}           %xypic macro for latex2.09
\usepackage{bm}
\usepackage{bbm}

\usepackage{bbding}
\usepackage{txfonts}
\usepackage{amscd}
\usepackage[all]{xy}
\usepackage[shortlabels]{enumitem}
\usepackage{ifpdf}
\ifpdf
  \usepackage[colorlinks,final,backref=page,hyperindex]{hyperref}
\else
  \usepackage[colorlinks,final,backref=page,hyperindex,hypertex]{hyperref}
\fi
\usepackage{tikz}
\usepackage[active]{srcltx}
\usepackage{caption}

\usepackage{comment}    %% comment out a paragraph
\makeatletter

\newtheorem{thm}{Theorem}[section]
\newtheorem{lem}[thm]{Lemma}
\newtheorem{cor}[thm]{Corollary}
\newtheorem{pro}[thm]{Proposition}
\theoremstyle{definition}
\newtheorem{ex}[thm]{Example}
\newtheorem{rmk}[thm]{Remark}
\newtheorem{defi}[thm]{Definition}

\newcommand {\emptycomment}[1]{}

\newcommand{\A}{{\mathbb A}}

\newcommand{\lon }{\,\rightarrow\,}
\newcommand{\be }{\begin{equation}}
\newcommand{\ee }{\end{equation}}

\newcommand{\Br}{\mathrm {Br}}

\newcommand{\h}{\mathfrak h}

\newcommand{\G}{\mathbb G}
\newcommand{\GGG}{\mathscr{ G}}
\newcommand{\HHH}{\mathscr{ H}}
\newcommand{\HH}{\mathbb H}
\newcommand{\R}{\mathbb R}

\newcommand{\huaB}{\mathcal{B}}%{{\mathcal{E}}}%{\mathcal{B}}

\newcommand{\huaY}{\mathcal{Y}}

\newcommand{\huaP}{\mathcal{P}}
\newcommand{\frkF}{\mathfrak F}
\newcommand{\frkG}{\mathfrak G}

\newcommand{\Id}{{\rm{Id}}}

\newcommand{\br}[1]{   [ \cdot,    \cdot  ]   }

\newcommand{\End}{\mathrm{End}}

\newcommand{\Perm}{\mathrm{Bij}}
\newcommand{\Map}{\mathrm{Map}}

\newcommand{\pr}{\mathrm{pr}}
\newcommand{\inv}{\mathrm{inv}}

\begin{document}

\title[Braided dynamical groups and dynamical Yang-Baxter equation]{Braided dynamical groups, the dynamical Yang-Baxter equation and related structures}

\author{Chengming Bai}
\address{Chern Institute of Mathematics and LPMC, Nankai University,
Tianjin 300071, China}
\email{baicm@nankai.edu.cn}

\author{Li Guo}
\address{Department of Mathematics and Computer Science,
         Rutgers University,
         Newark, NJ 07102}
\email{liguo@rutgers.edu}

\author{Yunhe Sheng}
\address{Department of Mathematics, Jilin University, Changchun 130012, Jilin, China}
\email{shengyh@jlu.edu.cn}

\author{You Wang}
\address{Department of Mathematics, Jilin University, Changchun 130012, Jilin, China}
\email{wangyou20@mails.jlu.edu.cn}

%\date{\today}

\begin{abstract}
We introduce the notion of a braided dynamical group which is a matched pair of dynamical groups satisfying extra conditions. It is shown to give a solution of the dynamical Yang-Baxter equation and at the same time a braided groupoid, thereby integrating the approaches of Andruskiewitsch and Matsumoto-Shimizu respectively that use these two notions to produce quiver-theoretical solutions of the Yang-Baxter equation.
We pursue this connection further by relative Rota-Baxter operators on dynamical groups, which give rise to matched pairs of dynamical groups. As the derived structures of relative Rota-Baxter operators on dynamical groups, dynamical post-groups are introduced and are shown to be equivalent to braided dynamical groups. Finally, skew-braces are generalized to dynamical skew-braces as another equivalent notion of braided dynamical groups.
\end{abstract}

\renewcommand{\thefootnote}{}

\footnotetext{2020 Mathematics Subject Classification.
16T25, %Yang-Baxter equations 
17B38, %Yang-Baxter equations and Rota-Baxter operators
18M15, %Braided monoidal categories and ribbon categories
20N02 %Sets with a single binary operation (groupoids)
%16T20, % Ring-theoretic aspects of quantum groups,
%81R50  %Quantum groups and related algebraic methods applied to problems in quantum theory
}

\keywords{braided dynamical group, dynamical Yang-Baxter equation, Rota-Baxter operator, post-group, skew brace}

\maketitle

\vspace{-1cm}

\tableofcontents

\vspace{-1cm}

\allowdisplaybreaks

\section{Introduction}

In the seminal work \cite{LYZ}, Lu-Yan-Zhu showed that braided groups give rise to solutions of the Yang-Baxter equation. The purpose of this paper is to generalize this result to the dynamical setting, which allows to integrate Andruskiewitsch's construction of quiver-theoretical solutions of the Yang-Baxter equation by braided groupoids with another construction by Matsumoto and Shimizu.

\subsection{Set-theoretical solutions of the Yang-Baxter equation and related structures}

In 1967, Yang derived the Yang-Baxter equation as the consistency condition for factorization in the quantum many-body problem \cite{Yang}. In 1972, the same equation emerged independently in Baxter's analysis of the eight-vertex model \cite{Baxter}. A solution to the Yang-Baxter equation on a vector space \( V \) is an invertible linear map $ R : V \otimes V \to V \otimes V $ satisfying
$$ (R \otimes \mathrm{Id}_V)(\mathrm{Id}_V \otimes R)(R \otimes \mathrm{Id}_V) = (\mathrm{Id}_V \otimes R)(R \otimes \mathrm{Id}_V)(\mathrm{Id}_V \otimes R). $$
The Yang-Baxter equation plays an essential role in broad areas of mathematics and mathematical physics, including integrable systems, quantum groups, Hopf algebras, knot theory, statistical mechanics, quantum field theory and braided categories.

In \cite{Drinfeld}, Drinfeld proposed to study set-theoretical solutions of the Yang-Baxter equation, which later on have emerged as an actively pursued research area, attracting considerable attention from mathematicians and mathematical physicists. Etingof-Schedler-Soloviev \cite{Etingof1} focused on the non-degenerate involutive set-theoretical solutions and the associated structure groups which can be used to classify the set-theoretical solutions. Lu-Yan-Zhu \cite{LYZ} showed that braided groups (as well as  compatible actions and braiding operators on groups) give rise to solutions of the Yang-Baxter equation, and a set-theoretical solution of the Yang-Baxter equation on a set $X$ can be generalized to a set-theoretical solution on its associated structure group $G_X$.

Rump \cite{Rump3,Rump2} introduced the notion of braces as a generalization of radical rings, and showed that braces are equivalent to linear cycle sets and give rise to non-degenerate involutive set-theoretical solutions of the Yang-Baxter equation. Guarnieri and Vendramin \cite{GV} introduced the notion of skew braces as a nonabelian generalization of braces and proved that skew braces produce non-degenerate (not involutive) set-theoretical solutions in different ways. In \cite{Gubarev}, Bardakov and Gubarev
established the relation between skew braces and Rota-Baxter groups introduced in \cite{GLS}. As the derived structure of Rota-Baxter operators on groups, the notion of post-groups was introduced in  \cite{BGST}, which are equivalent to skew braces and thus provide set-theoretical solutions of the Yang-Baxter equation.
\vspace{-.2cm}
\subsection{The dynamical Yang-Baxter equation}

As a significant generalization of the Yang-Baxter equation, the quantum dynamical Yang-Baxter equation was first introduced by Gervais and Neveu in physics \cite{GN} and then studied by Felder \cite{F1,F2} from a mathematical viewpoint. Let $\h$ be a finite-dimensional abelian Lie algebra over $\mathbb{C}$ and $V$ be a semisimple finite-dimensional $\h$-module. Then the quantum dynamical Yang-Baxter equation \cite{Etingof4} with respect to a meromorphic function $R:\h^*\to \End(V \otimes V)$ is given as follows:
$$
R^{12}(\lambda-\gamma h^{(3)})R^{13}(\lambda) R^{23}(\lambda-\gamma h^{(1)})=R^{23}(\lambda)  R^{13}(\lambda-\gamma h^{(2)}) R^{12}(\lambda),
$$
where $\gamma$ is a complex number, $R^{12}(\lambda-\gamma h^{(3)})(v_1\otimes v_2 \otimes v_3):=(R^{12}(\lambda-\gamma \mu)(v_1\otimes v_2))\otimes v_3$ if $v_3$ has weight $\mu$, and
$R^{13}(\lambda-\gamma h^{(2)}),R^{23}(\lambda-\gamma h^{(1)})$ are defined analogously.
A quantum dynamical $R$-matrix is an invertible solution of this equation satisfying the $\h$-invariant condition.
Etingof and Schiffmann \cite{Etingof2} found out that exchange matrices are solutions of the quantum dynamical Yang-Baxter equation using the exchange construction. Etingof and Varchenko also considered the classical dynamical Yang-Baxter equation   as the quasiclassical analogue of the quantum dynamical Yang-Baxter equation \cite{Etingof3}. The quantum dynamical Yang-Baxter equation and the classical dynamical Yang-Baxter equation appeared in many areas of mathematics and mathematical physics, such as conformal field theory \cite{F1}, dynamical quantum groups \cite{Etingof4,KS}, integrable systems \cite{F1,SZ}, Lie superalgebras \cite{Kar}, moduli spaces of flat connections \cite{Xu} and representation theory \cite{Etingof4,Sto}.

Motivated by Drinfeld's proposal to study set-theoretical solutions of the Yang-Baxter equation, it is natural to study the set-theoretical analogue of the quantum dynamical Yang-Baxter equation. Shibukawa \cite{Shibukawa1,Shibukawa2,Shibukawa3} formulated the dynamical Yang-Baxter equation:
\vspace{-.1cm}
$$ R(\lambda)_{12} R(\phi_X(\lambda,X^{1}))_{23} R(\lambda)_{12}=R(\phi_X(\lambda,X^{1}))_{23}R(\lambda)_{12}R(\phi_X(\lambda,X^{1}))_{23}. $$
See Definition \ref{defi:dYBE} for more details. In fact, the Yang-Baxter equation can be defined in any monoidal category. The usual Yang-Baxter equation is defined in the monoidal category of vector spaces. Note that the category of dynamical sets ${\bf DSet_{\Lambda}}$ is also a monoidal category, and the above dynamical Yang-Baxter equation is exactly the Yang-Baxter equation defined in the monoidal category ${\bf DSet_{\Lambda}}$. Up to now, dynamical Yang-Baxter equation and its set-theoretical solutions have been studied extensively.
In \cite{ST10}, Shibukawa and Takeuchi studied an analogue of FRT construction in the dynamical setting.
Shibukawa also discovered that Hopf algebroids and rigid tensor categories can be constructed using dynamical Yang-Baxter maps \cite{Shi16}.
Matsumoto introduced the notion of dynamical braces, which can give rise to unitary solutions of the dynamical Yang-Baxter equation \cite{Matsumoto}.

Andruskiewitsch considered the Yang-Baxter equation in the monoidal category of quivers ${\bf Quiv_{\Lambda}}$ in \cite{Andruskiewitsch}, and proved that braided groupoids give rise to quiver-theoretical solutions of the Yang-Baxter equation. Recently, the notion of post-groupoids was introduced in \cite{STZ}, which are equivalent to braided groupoids, and also give rise to quiver-theoretical solutions of the Yang-Baxter equation.
There is another construction of quiver-theoretical solutions of the Yang-Baxter equation. In \cite{MS}, Matsumoto and Shimizu found an embedding monoidal functor $Q$ from the monoidal category ${\bf DSet_{\Lambda}}$ of dynamical sets to the monoidal category  ${\bf Quiv_{\Lambda}}$ of quivers. They further showed that a solution of the dynamical Yang-Baxter equation produces a quiver-theoretical solution of the Yang-Baxter equation on the induced quiver.
\vspace{-.2cm}
\subsection{Main results and outline of the paper}
The first purpose of this paper is to introduce the notion of a braided dynamical group as a special matched pair of dynamical groups (Definition \ref{def-braid-d-group}), with an application to
integrate the above two approaches to the construction of quiver-theoretical solutions of the Yang-Baxter equation.
On the one hand, this notion allows us to generalize the braided groups approach to the Yang-Baxter equation developed in the seminal work \cite{LYZ} to the dynamical setting.
More precisely, we show that a braided dynamical group gives rise to a solution of the dynamical Yang-Baxter equation (Theorem~\ref{braided-solution}). On the other hand,  we show that the aforementioned functor $Q:{\bf DSet_{\Lambda}}\to {\bf Quiv_{\Lambda}}$ sends dynamical groups to groupoids. We further generalize this construction to the braided setting and show that the functor $Q$ sends braided dynamical groups to braided groupoids (Theorem~\ref{braided-d-group-braided-gpd}). Therefore, the two  approaches to obtain a quiver-theoretical solution of the Yang-Baxter equation, one by Andruskiewitsch's construction from braided groupoids \cite{Andruskiewitsch} and another one by Matsumoto-Shimizu's construction from solution of the dynamical Yang-Baxter equation \cite{MS}, can both be viewed under the lens of braided dynamical groups (Corollary~\ref{c:twoconst}).

Given the critical role of braided dynamical groups in solving the various Yang-Baxter equations, our second purpose is to further understand this notion, by generalizing to the dynamical setting several well-known structures, including (skew) braces, Rota-Baxter operators and post-groups. We start with the notion of relative Rota-Baxter operators on dynamical groups, which lead to the  factorization of dynamical groups. We prove that relative Rota-Baxter operators on dynamical groups lead to matched pairs of dynamical groups (Proposition~\ref{rRBO-mp}). We then introduce the notion of dynamical post-groups, which are equivalent to braided dynamical groups (Theorem \ref{DPG-BDG} and Proposition \ref{BDG-DPG}) and closely related to relative Rota-Baxter operators on dynamical groups (Propositions \ref{rRBO-d-post-group} and \ref{Id-rRBO}).
%\li{Li: Add the relation between dynamical post-groups and braided dynamical groups as shown in the diagram below?}
%. A relative Rota-Baxter operator on a dynamical group naturally splits the dynamical group structure into a dynamical post-group structure (Proposition~\ref{rRBO-d-post-group}). Conversely, a dynamical post-group gives rise to a relative Rota-Baxter operator on its sub-adjacent dynamical group (Proposition~\ref{Id-rRBO}).
We finally study dynamical skew-braces and show their equivalence to dynamical post-groups (Propositions~\ref{dpg-dsb} and \ref{dsb-dpg}), and thereby braided dynamical groups.

The main notions and results of the paper can be summarized in the following diagram where the bold-faced terms and dashed arrows are newly obtained.
\vspace{-.1cm}
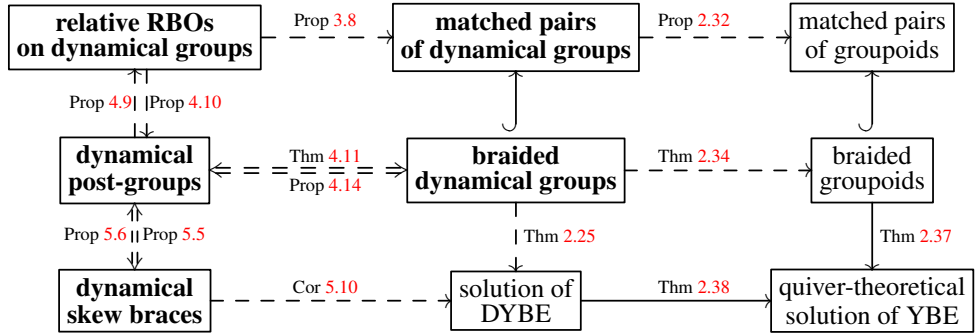
\begin{equation*}
\vcenter{\hbox{\footnotesize
\xymatrix@C=5.0em@R=2.5em{
% New first row
*+[F]\txt{\footnotesize \bf relative RBOs \\ \bf on dynamical groups} \ar@{-->}[r]^{\text{Prop~\ref{rRBO-mp}}}
\ar@<1.0ex>@{-->}[d]_{\text{Prop~\ref{rRBO-d-post-group}}~~} &
*+[F]\txt{\footnotesize \bf matched pairs \\ \bf of dynamical groups}   \ar@{-->}[r]^{\text{Prop~\ref{mp-groupoid-double}}}
\ar@{<-_{)}}[d(0.65)] &
*+[F]\txt{\footnotesize matched pairs \\ of groupoids}
\ar@{<-_{)}}[d(0.65)] \\
% Original first row (now second row)
*+[F]\txt{\footnotesize {\bf dynamical}\\ \bf post-groups} \ar@{-->}[u]_{~~\text{Prop~\ref{Id-rRBO}}}
\ar@{<==>}[r]^{\text{Thm~\ref{DPG-BDG}}}_{\text{Prop~\ref{BDG-DPG}}}
\ar@{<==>}[d]^{\text{Prop~\ref{dpg-dsb}}}_{\text{Prop~\ref{dsb-dpg}}}
& *+[F]\txt{\footnotesize \bf braided\\ \bf dynamical groups}   %\ar@{<--<}[d(0.6)]
  \ar@{-->}[r]^{\text{Thm~\ref{braided-d-group-braided-gpd}}}
  \ar@{-->}[d]^{\text{Thm~\ref{braided-solution}}}
& *+[F]\txt{\footnotesize braided\\ groupoids}
%\ar@{_{(}->}[u]
  \ar[d]^{\text{Thm \ref{QYBE}}} \\
% Original second row (now third row)
*+[F]\txt{\footnotesize  \bf dynamical \\ \bf skew braces}
  \ar@{-->}[r]^{\text{Cor~\ref{d-skew-brace-solution}}}
& *+[F]\txt{\footnotesize solution of\\DYBE}
  \ar[r]^{\text{Thm~\ref{br-dset-br-quiver}}}
& *+[F]\txt{\footnotesize quiver-theoretical \\solution of YBE}
}}}
\end{equation*}

The paper is organized as follows. In Section \ref{sec:braided-d-groups}, we introduce braided dynamical groups, which lead to solutions of the dynamical Yang-Baxter equation. In Section \ref{sec:relative-RBO-mp}, we give the notion of relative Rota-Baxter operators on dynamical groups, which give rise to matched pairs of dynamical groups. In Section \ref{sec:d-post-group}, as the derived structure of relative Rota-Baxter operators on dynamical groups, the dynamical post-groups are defined, which are equivalent to braided dynamical groups. In Section \ref{sec:d-skew-brace}, we introduce the notion of dynamical skew-braces as a generalization of skew braces to the dynamical setting, and show that dynamical skew-braces are equivalent to dynamical post-groups, as well as braided dynamical groups.

\vspace{-.2cm}

\section{Matched pairs of dynamical groups, braided dynamical groups and the dynamical Yang-Baxter equation}\label{sec:braided-d-groups}

In this section, we first recall the notions of dynamical sets and the dynamical Yang-Baxter equation. Then we introduce the notion of a matched pair of dynamical groups. A braided dynamical group is a special matched pair of dynamical groups, and produces a solution of the dynamical Yang-Baxter equation. Moreover, we show that braided dynamical groups can give rise to braided groupoids, which are associated to the quiver-theoretical solutions of the Yang-Baxter equation.

\subsection{Dynamical sets and dynamical Yang-Baxter equation}

Throughout this paper, we denote by $\Lambda$ a fixed set of dynamical parameters. We start with a brief review of the category ${\bf DSet_{\Lambda}}$ of dynamical sets over $\Lambda$, which is a monoidal category introduced in \cite{MS,Rump,Shibukawa3}. A {\bf dynamical set} is a set $X$ together with a {\bf structure map} $\phi_X: \Lambda \times X \to \Lambda $. A dynamical set will be denoted by $\mathscr{X}=(X,\Lambda,\phi_X)$. The component maps $\phi_X(\lambda,-):X\to \Lambda$ are denoted by $\phi_X^\lambda$. $\mathscr{S}=(S,\Lambda,\phi_{S})$ is called a {\bf dynamical subset} of $\mathscr{X}$, if $S$ is a subset of $X$ and $\phi_{S}$ satisfies $\phi_{S}(\lambda,a)=\phi_{X}(\lambda,a)$ for all $\lambda\in \Lambda,a \in S$. A dynamical set
$\mathscr{X}=(X,\Lambda,\phi_X)$ is called {\bf constant}, if $\phi_X(\lambda,x)=\lambda$ for all $\lambda\in \Lambda,x\in X$. A morphism of dynamical sets from $\mathscr{X}=(X,\Lambda,\phi_X)$ to $\mathscr{Y}=(Y,\Lambda,\phi_Y)$ is  a map $f:\Lambda  \times X \to Y $ satisfying
\begin{equation}\label{morphism-d-set}
\phi_Y(\lambda,f(\lambda,x))=\phi_X(\lambda,x),\quad \forall \lambda\in \Lambda, x\in X.
\end{equation}
 In other words, a morphism is a collection of maps $f_\lambda:X\to Y$ with $f_\lambda=f(\lambda,-)$ which makes the following diagram commutative for all $\lambda\in \Lambda$:
\[
        \vcenter{\xymatrix{
  X \ar[d]_{\phi_X^\lambda} \ar[r]^{f_\lambda}
                & Y \ar[d]^{\phi_Y^\lambda}  \\
   \Lambda  \ar[r]_{\Id}
                &  \Lambda
                }}
    \]
The composition $g \circ f: \Lambda  \times X \to Z $ of morphisms $f:\Lambda  \times X \to Y $ and $g:\Lambda  \times Y \to Z $ is defined by
$$(g \circ f)(\lambda,x):=g(\lambda,f(\lambda,x)), \quad \forall \lambda\in \Lambda,$$
which is also expressed as
\begin{eqnarray}\label{composition-maps}
 (g \circ f)_{\lambda}=g_\lambda \circ f_\lambda.
\end{eqnarray}

The category ${\bf DSet_{\Lambda}}$ has a terminal object, given by $\Lambda$ together with the second projection $\phi_{\Lambda}=\Pr_2:\Lambda \times \Lambda \to \Lambda$  as the structure map.

Given two objects $\mathscr{X}=(X,\Lambda,\phi_X),\mathscr{Y}=(Y,\Lambda,\phi_Y) \in {\bf DSet_{\Lambda}}$, we define their tensor product
$$ \mathscr{X} \otimes \mathscr{Y}:=(X\times Y,\Lambda,\phi_{X\times Y}) \in {\bf DSet_{\Lambda}},$$
where $X\times Y$ is the Cartesian product and the structure map $\phi_{X\times Y}:\Lambda \times X\times Y \to \Lambda$ is given by
\begin{equation}\label{phi-tensor}
\phi_{X\times Y}(\lambda,x,y):=\phi_Y (\phi_X(\lambda,x),y), \quad \forall \lambda \in \Lambda, x\in X, y\in Y.
\end{equation}
For morphisms $f:\Lambda \times X \to X'$ and $g:\Lambda \times Y \to Y'$ in ${\bf DSet_{\Lambda}}$, their tensor product $f\otimes g:\Lambda \times (X \times Y) \to X' \times Y'$ is given by
$$ (f\otimes g)(\lambda,x,y)=\Big(f(\lambda,x),g(\phi_X(\lambda,x),y) \Big). $$
The singleton ${\bf 1}:=\{ 1 \}$ is an object of ${\bf DSet_{\Lambda}}$ with the structure map $\phi_{\bf 1}(\lambda,1)=\lambda$ for all $\lambda \in \Lambda.$ There are natural isomorphisms
$$ a_{X,Y,Z}: (\mathscr{X}\otimes \mathscr{Y}) \otimes \mathscr{Z} \to \mathscr{X} \otimes (\mathscr{Y} \otimes \mathscr{Z}),\quad l_X: ({\bf 1}\otimes \mathscr{X})\to \mathscr{X}, \quad r_X:(\mathscr{X} \otimes{\bf 1})  \to \mathscr{X} $$
for objects $\mathscr{X}=(X,\Lambda,\phi_X),\mathscr{Y}=(Y,\Lambda,\phi_Y),\mathscr{Z}=(Z,\Lambda,\phi_Z)\in {\bf DSet_{\Lambda}}$ defined by
$$  a_{X,Y,Z}(\lambda,(x,y),z)=(x,(y,z)),\quad  l_X(\lambda,(1,x))=x,\quad r_X(\lambda,(x,1))=x,$$
for all $\lambda\in \Lambda, x\in X, y\in Y, z\in Z$. Therefore, the category ${\bf DSet_{\Lambda}}$ is a monoidal category with the tensor product $\otimes$, the unit object ${\bf 1}$ and the natural isomorphisms defined as above.

\begin{rmk}
Dynamical sets can be viewed as a discrete dynamical systems. By regarding  $\Lambda$ as the given system, the structural map $\phi_X:\Lambda \times X \to \Lambda$ constitutes a family of iterative functions $\{\phi_x:\Lambda \to \Lambda\}_{x\in X}$ that capture the evolution of the system.
\end{rmk}

Consider the braided equation and braided objects in the monoidal category ${\bf DSet_{\Lambda}}$, one obtains the following notion of dynamical Yang-Baxter equation, originally given in \cite{Shibukawa1}.

\begin{defi}\label{defi:dYBE}
A solution of the {\bf dynamical Yang-Baxter equation} on a dynamical set $\mathscr{X}=(X,\Lambda,\phi_X)$ is an isomorphism of dynamical sets $R: \mathscr{X} \otimes \mathscr{X} \to \mathscr{X} \otimes \mathscr{X} $ satisfying:
\begin{eqnarray}\label{DYBE}
R(\lambda)_{12} R(\phi_X(\lambda,X^{1}))_{23} R(\lambda)_{12}=R(\phi_X(\lambda,X^{1}))_{23}R(\lambda)_{12}R(\phi_X(\lambda,X^{1}))_{23},
\end{eqnarray}
where $R(\lambda)_{12},R(\phi_X(\lambda,X^{1}))_{23}:X \times X \times X  \to X \times X \times X $ are given by
\begin{eqnarray*}
R(\lambda)_{12}(x,y,z)=(R(\lambda,x,y),z), \quad R(\phi_X(\lambda,X^{1}))_{23}(x,y,z)=(x,R(\phi_X(\lambda,x),y,z)),
\end{eqnarray*}
for all $ \lambda\in \Lambda, x,y,z\in X.$ We denote $R(\lambda,x,y)=(\varphi^{\lambda}_{x}(y),\psi^{\lambda}_{y}(x))$. The solution $R$ is called {\bf non-degenerate} if $\varphi^{\lambda}_{x}$ and $\psi^{\lambda}_{y}$ are bijective for all $\lambda \in \Lambda,x,y\in X$.

A dynamical set $\mathscr{X}$ with a solution $R$ of the dynamical Yang-Baxter equation   is called a {\bf braided dynamical set} and denoted by $(\mathscr{X} ,R)$. Denote by $\Br({\bf DSet_{\Lambda}})$ the category of braided dynamical sets over $\Lambda$.

\end{defi}

\begin{rmk}
By \eqref{morphism-d-set} and \eqref{phi-tensor}, $R$ is a morphism of dynamical sets from $\mathscr{X} \otimes \mathscr{X} $ to $\mathscr{X} \otimes \mathscr{X} $ if and only if the following condition is satisfied:
\begin{eqnarray}\label{weight-zero}
\phi_X (\phi_X(\lambda,\varphi^{\lambda}_{x}(y)),\psi^{\lambda}_{y}(x))=
\phi_X (\phi_X(\lambda,x),y), \quad \forall\lambda\in \Lambda,x,y\in X,
\end{eqnarray}
which is called {\bf the weight zero condition} {\rm(\cite{Rump})} or {\bf invariant condition} {\rm(\cite{Shibukawa2})}.
\end{rmk}

\begin{rmk}
Note that $R$ is an isomorphism if and only if the components $R(\lambda):X \times X \to X \times X$ are bijective for all $\lambda \in \Lambda$; while in the original definition given in \cite{MS,Shibukawa1}, $R$ is not required to be an isomorphism.
\end{rmk}

\subsection{Matched pairs of dynamical groups and braided dynamical groups}

We first recall the notion of a dynamical group, which is a group object in the monoidal category ${\bf DSet_{\Lambda}}$, as an important ingredient in the definition of braided dynamical groups.

\begin{defi}{\rm(\cite{Rump})}\label{dynamical-group}
A {\bf dynamical group} structure on a dynamical set $\mathscr{G}=(G, \Lambda,\phi_{G})$ consists of the following data:
\begin{itemize}
\item[{\rm(i)}] a morphism (called the {\bf product}) $\circ:\mathscr{G}\otimes \mathscr{G}\to \mathscr{G}$ such that the following associative law holds:
\begin{equation}\label{product}
  a \circ_\lambda (b \circ_{\phi_{G}(\lambda,a)} c)=
  (a \circ_\lambda b) \circ_\lambda c, \quad \forall \lambda\in \Lambda, a,b,c\in G,
\end{equation}
where  $\circ_\lambda:G\times G\to G$ is the components of $\circ$;
\item[{\rm(ii)}]  an element $e_G\in G$ (called the {\bf unit}) such that
\begin{equation}\label{unit}
a \circ_\lambda e_G=e_G \circ_\lambda a=a,\quad \phi(\lambda,e_G)=\lambda,\quad\forall \lambda\in \Lambda, a\in G;
\end{equation}
\item[{\rm(iii)}] for all $\lambda\in \Lambda, a\in G$, there always exists $\bar{a}^\lambda\in G$ (called the {\bf inverse} of $a$ with respect to the product $\circ_\lambda$) such that $$a \circ_\lambda \bar{a}^\lambda=\bar{a}^\lambda \circ_{\phi_{G}(\lambda,a)} a =e_G.$$
\end{itemize}

 We will denote a dynamical group by $(\mathscr{G},\circ)$.
 If $(\mathscr{G},\circ)$ only satisfies the associative law {\rm(i)}, then it is called a {\bf dynamical semi-group}.
\end{defi}

\begin{rmk}
\begin{itemize}
\item[{\rm(i)}] The associative law in Definition \ref{dynamical-group} comes from the notion of group objects in the monoidal category ${\bf DSet_{\Lambda}}$.
\item[{\rm(ii)}] For a dynamical group $(G, \Lambda,\phi_{G},\{\circ_\lambda\}_{\lambda\in \Lambda})$,  the unit $e_G$ is unique.
\item[{\rm(iii)}] For any $\lambda\in \Lambda, a\in G$, the inverse element $\bar{a}^\lambda$ is also uniquely determined according to the associative law  in Definition \ref{dynamical-group}.
\end{itemize}
\end{rmk}

\begin{rmk}
Since the product $\circ: \mathscr{G} \otimes \mathscr{G} \to \mathscr{G}$ is a morphism from $\mathscr{G} \otimes \mathscr{G}$ to $\mathscr{G}$ in the category ${\bf DSet_{\Lambda}}$,  by \eqref{morphism-d-set} and \eqref{phi-tensor}, it follows that
\begin{eqnarray}\label{phi-asso}
\phi_{G}(\lambda,a\circ_{\lambda} b) =\phi_{G}(\phi_{G}(\lambda,a),b), \quad \forall \lambda\in \Lambda,a,b\in G.
\end{eqnarray}

Note that the condition \eqref{phi-asso} is also deduced in \cite[Proposition 3.4]{Matsumoto} using a different approach.
\end{rmk}

\begin{ex}\label{ex-d-group-R}
Let $(\mathbb{R},+,\cdot)$ be the field of real numbers. Then $(\mathbb{R},\Lambda=\mathbb{R},\phi)$ is a dynamical set, where  $\phi:\mathbb{R} \times \mathbb{R}\to \mathbb{R}$ is given by
$$
\phi(\lambda,a)=\lambda \cdot (\lambda a+1).
$$
Define the product $\circ_{\lambda}:\mathbb{R} \times \mathbb{R} \to \mathbb{R}$ by
\begin{eqnarray*}
a \circ_{\lambda} b&=&a+(\lambda a+1)^{2} \cdot b,
\end{eqnarray*}
for all $\lambda,a,b\in \mathbb{R}$.
Then, we have
\begin{eqnarray*}
(a \circ_{\lambda} b) \circ_{\lambda} c &=& a+(\lambda a+1)^{2} \cdot b+(\lambda a+\lambda(\lambda a+1)^{2}\cdot b+1)^2 \cdot c\\
&=& a+(\lambda a+1)^{2}\cdot (b+(\lambda^2 ab+\lambda b+1)^2 \cdot c)\\
&=& a \circ_{\lambda} (b \circ_{\phi(\lambda,a)} c),
\end{eqnarray*}
which implies that $\circ_{\lambda}$ satisfies the associative law {\rm(i)} in Definition \ref{dynamical-group}. Moreover, we can check that $0$ is the unit and $a^{\lambda}=-\frac{a}{(\lambda a+1)^{2}}$ is the inverse of $a$ with respect to $\circ_{\lambda}$ for all $\lambda,a \in \mathbb{R}.$ Therefore,  $(\mathbb{R},\Lambda=\mathbb{R},\phi,\{\circ_{\lambda}\}_{\lambda\in\R})$ is a dynamical group.
\end{ex}

\begin{ex}\label{ex-d-group}
Let $G=\{0,1,2\}$ be a set containing three elements. Then $(G,\Lambda=\{\lambda_1,\lambda_2,\lambda_3\},\phi)$ is a dynamical set, where $\phi:\Lambda \times G \to \Lambda$ is given by
\begin{center}
\begin{minipage}{0.25\textwidth}
\centering
\(
\begin{array}{c|ccc}
    \phi & 0 & 1 & 2 \\ \hline
    \lambda_1 & \lambda_1 & \lambda_3 & \lambda_2 \\
    \lambda_2 & \lambda_2 & \lambda_3 & \lambda_1 \\
    \lambda_3 & \lambda_3 & \lambda_1 & \lambda_2 \\
\end{array}
\)
\par\vspace{0.5ex}
\end{minipage}%
\end{center}
Define the product $\circ_{\lambda_i}:G \times G \to G$ as follows.
\begin{center}
\begin{minipage}{0.25\textwidth}
\centering
\(
\begin{array}{c|ccc}
    \circ_{\lambda_1} & 0 & 1 & 2 \\ \hline
    0 & 0 & 1 & 2 \\
    1 & 1 & 0 & 2 \\
    2 & 2 & 1 & 0 \\
\end{array}
\)
\end{minipage}%
\hspace{0.5em}% 减小水平间距
\begin{minipage}{0.25\textwidth}
\centering
\(
\begin{array}{c|ccc}
    \circ_{\lambda_2} & 0 & 1 & 2 \\ \hline
    0 & 0 & 1 & 2 \\
    1 & 1 & 2 & 0 \\
    2 & 2 & 1 & 0 \\
\end{array}
\)
\end{minipage}%
\hspace{0.5em}% 减小水平间距
\begin{minipage}{0.25\textwidth}
\centering
\(
\begin{array}{c|ccc}
    \circ_{\lambda_3} & 0 & 1 & 2 \\ \hline
    0 & 0 & 1 & 2 \\
    1 & 1 & 0 & 2 \\
    2 & 2 & 0 & 1 \\
\end{array}
\)
\end{minipage}
\end{center}
By a direct verification, we check that $(G=\{0,1,2\}, \Lambda=\{\lambda_1,\lambda_2,\lambda_3\},\phi,\{\circ_{\lambda_i}\}_{\lambda_i\in \Lambda})$ is a dynamical group.
\end{ex}

\begin{defi}
Let $(\mathscr{G},\circ)$ be a dynamical group. A dynamical subset $\mathscr{S}=(S,\Lambda,\phi_S)\subset \mathscr{G}$ is  called a {\bf dynamical subgroup}  if  the following conditions are satisfied.
\begin{itemize}
\item[{\rm(i)}] For all $\lambda\in \Lambda$ and $ a,b\in S,$  we have $ a \circ_{\lambda} b \in S$.
\item[{\rm(ii)}] $e_G \in S$, where $e_G$ is the unit of the dynamical group  $(\mathscr{G},\circ)$.
\item[{\rm(iii)}] For all $\lambda\in \Lambda, a\in S$, the inverse $\bar{a}^\lambda$ is in $S$.
\end{itemize}
\end{defi}

A {\bf homomorphism of dynamical groups} from $(\mathscr{G},\circ)$ to $(\mathscr{H},\circ')$ is a morphism of dynamical sets $\Psi:\mathscr{G} \to \mathscr{H}$ such that
\begin{eqnarray}\label{homo-d-group}
\Psi_{\lambda}(a \circ_\lambda b)= \Psi_{\lambda}(a)\circ'_{\lambda}\Psi_{\phi_G(\lambda,a)}
(b)=\Psi_{\lambda}(a)\circ'_{\lambda}\Psi_{\phi_{H}(\lambda,\Psi(\lambda,a))}(b), \quad \forall \lambda \in \Lambda, a,b\in G.
\end{eqnarray}
%In fact, the condition \eqref{homo-d-group} is equivalent to the condition $\Psi \circ m=m' \circ (\Psi \times \Psi) $ with the composition given by \eqref{composition-maps}.

Now we introduce the notion of a constant dynamical group.

\begin{defi}
A dynamical group $(\mathscr{H},\cdot)$ is called a {\bf constant dynamical group} if $$\phi_H(\lambda,x)=\lambda,\quad \forall \lambda\in \Lambda, x\in H.$$
\end{defi}

We usually denote by $x^{\lambda}$ the inverse of $x\in H$ with respect to the product $\cdot_\lambda$ in a constant dynamical group $(\mathscr{H},\cdot)$, while we denote by $\bar{a}^{\lambda}$ the inverse of $a\in G$ with respect to the product $\circ_\lambda$ in an arbitrary dynamical group $(\mathscr{G},\circ)$.

\begin{rmk}
Let $(\mathscr{H},\cdot)$ be a constant dynamical group. Since $\phi_H(-,x)=\Id_{\Lambda}$ for all $x\in H$, by \eqref{product}, we have
$x \cdot_\lambda (y \cdot_\lambda z)=(x \cdot_\lambda y) \cdot_\lambda z$    for all
$ \lambda\in \Lambda, x,y,z\in H, $ which implies that a constant dynamical group can be viewed as a collection of groups on the same set  indexed by $\lambda\in \Lambda$ with their units identified.
 %In particular, a trivial dynamical group can be viewed as a collection of isomorphic groups indexed by $\lambda\in \Lambda$ with their units identified.
\label{r:constsyngp}
\end{rmk}

\begin{defi}
A constant dynamical group $(\mathscr{H},\cdot)$ is called a {\bf trivial dynamical group} if groups $(H,\cdot_\lambda)$ are isomorphic for all $\lambda\in \Lambda$.
\end{defi}

\begin{ex}\label{group-trivial}
Every group $(H,+)$ gives rise to a trivial dynamical group $(\mathscr{H},\cdot)$ with $\cdot_\lambda:=+$ for all $\lambda\in \Lambda.$
\end{ex}

\begin{ex}
Let $(\mathbb{R},+,\cdot)$ be the field of real numbers and $\Lambda=\{2k+1|k\in \mathbb{N}\}$.  Define $\cdot_\lambda:\mathbb{R} \times \mathbb{R} \to \mathbb{R}$ by
$$  a\cdot_\lambda b=\sqrt[\lambda]{a^{\lambda}+b^{\lambda} },\quad \forall \lambda\in \Lambda, a,b\in \mathbb{R}.$$
Define bijective maps $f_{\lambda}:\mathbb{R} \to \mathbb{R}$ by $f_{\lambda}(a)=a^{\lambda}$ for all $\lambda\in \Lambda$. Then we have
$$ f_{\lambda}(a\cdot_{\lambda} b)=a^\lambda+b^\lambda=f_{\lambda}(a)\cdot_{1} f_{\lambda}(b),\quad \forall a,b\in \mathbb{R},$$
which implies that $f_{\lambda}:(\mathbb{R},\cdot_{\lambda}) \to (\mathbb{R},\cdot_{1})$ are all isomorphisms of groups for all $\lambda\in \Lambda$. Therefore, $(\mathbb{R},\Lambda,\{\cdot_\lambda\}_{\lambda\in \Lambda})$ is a trivial dynamical group. However, the group structures $(\mathbb{R},\cdot_\lambda)$ are different for all $\lambda\in \Lambda$.
\end{ex}

Based on Remark~\ref{r:constsyngp}, we obtain the following examples.
\begin{ex}
Let $H$ be a set with four elements. Let $(H,\cdot_{\lambda_1})$ be the cyclic group $\mathbb{Z}_4$ and $(H,\cdot_{\lambda_2})$ be the Klein four-group $\mathbb{Z}_2\times \mathbb{Z}_2$. Then $(H, \Lambda=\{\lambda_1,\lambda_2\},\{\cdot_\lambda\}_{\lambda\in \Lambda})$ is a constant dynamical group, which is not trivial.
\end{ex}

\begin{ex}{\rm
Let $H$ be a set with eight elements. Then
$(H, \Lambda=\{\lambda_1,\lambda_2,\lambda_3,\lambda_4,\lambda_5\},\{\cdot_\lambda\}_{\lambda\in \Lambda})$ is a constant dynamical group, where
\begin{itemize}
\item[{\rm(i)}] $(H,\cdot_{\lambda_1})$ is the cyclic group $\mathbb{Z}_8$;
\item[{\rm(ii)}] $(H,\cdot_{\lambda_2})$ is the abelian group $\mathbb{Z}_4\times \mathbb{Z}_2$;
\item[{\rm(iii)}] $(H,\cdot_{\lambda_3})$ is the abelian group $\mathbb{Z}_2\times \mathbb{Z}_2\times \mathbb{Z}_2$;
\item[{\rm(iv)}] $(H,\cdot_{\lambda_4})$ is the dihedral group of order eight $\mathbb{D}_4=\langle r,s~|~r^4=s^2=e,~sr=r^{-1}s \rangle$, which is a non-abelian group;
\item[{\rm(v)}] $(H,\cdot_{\lambda_5})$ is the quaternion group $\mathbb{Q}_8=\langle i,j,k~|~i^2=j^2=k^2=ijk=-1\rangle$, which is a non-abelian group.
\end{itemize}
In general, given a classification of groups with a given order, it can induce a (non-trivial) constant dynamical group with the same order.
}\end{ex}

Recall from \cite{GV}  that a {\bf skew brace} $(H,\cdot_H,\circ_H)$ consists of a group $(H,\cdot_H)$ and a group $(H,\circ_H)$ such that
$$ a\circ_H (b \cdot_H c)=(a\circ_H b)\cdot_H a^{-1} \cdot_H (a\circ_H c),\quad \forall a,b,c\in H, $$
where $a^{-1}$ is the inverse of $a$ in the group $(H,\cdot_H)$. In particular, if $(H,\cdot_H)$ is an abelian group, then $(H,\cdot_H,\circ_H)$ is called a {\bf brace}.

\begin{ex}
Every (skew) brace $(H,\cdot_H,\circ_H)$ can be viewed as a  constant dynamical group $(H, \Lambda=\{\lambda_1,\lambda_2\},\{\cdot_\lambda\}_{\lambda\in \Lambda})$ with $\cdot_{\lambda_1}:=\cdot_H$ and $\cdot_{\lambda_2}:=\circ_H$.
\end{ex}

Let us introduce the notion of a matched pair of dynamical groups.
\begin{defi}
A {\bf matched pair of dynamical groups} is a triple $(\mathscr{G},\mathscr{H},\sigma)$, where $(\mathscr{G},\circ)$ and $(\mathscr{H},\cdot)$ are dynamical groups and
$$ \sigma:\Lambda \times G \times H \to H\times G,\quad (\lambda,a,x)\mapsto (a\overset{\lambda}{\rightharpoonup} x,a\overset{\phi_G(\lambda,a)}{\leftharpoonup\joinrel\relbar\joinrel\relbar\joinrel\relbar} x )   $$
is a morphism of dynamical sets from $\mathscr{G} \otimes \mathscr{H}$ to $\mathscr{H} \otimes \mathscr{G}$ satisfying the following conditions:
\begin{eqnarray}
\label{mp-d-groups-1} e_G \overset{\lambda}{\rightharpoonup} x&=&x;\\
\label{mp-d-groups-2} (a\circ_\lambda b)\overset{\lambda}{\rightharpoonup} x&=& a \overset{\lambda}{\rightharpoonup} (b\overset{\phi_G(\lambda,a)}{\relbar\joinrel\relbar\joinrel\rightharpoonup} x  );\\
\label{mp-d-groups-3} a \overset{\lambda}{\rightharpoonup} (x \cdot_{\phi_G(\lambda,a)} y)&=&
(a \overset{\lambda}{\rightharpoonup} x) \cdot_{\lambda} ( (a\overset{\phi_G(\lambda,a)}{\leftharpoonup\joinrel\relbar\joinrel\relbar\joinrel\relbar} x)         \overset{\phi_H(\lambda,a\overset{\lambda}{\rightharpoonup} x)} {\relbar\joinrel\relbar\joinrel\relbar\joinrel\relbar\joinrel\rightharpoonup}  y );\\
\label{mp-d-groups-4} a\overset{\phi_G(\lambda,a)}{\leftharpoonup\joinrel\relbar\joinrel\relbar\joinrel\relbar} e_H&=&a;\\
\label{mp-d-groups-5} a\overset{\phi_G(\lambda,a)}{\leftharpoonup\joinrel\relbar\joinrel\relbar\joinrel\relbar} (x \cdot_{\phi_G(\lambda,a)} y)&=&(a\overset{\phi_G(\lambda,a)}{\leftharpoonup\joinrel\relbar\joinrel\relbar\joinrel\relbar} x)\overset{\phi_H(\phi_G(\lambda,a),x)}{\leftharpoonup\joinrel\relbar\joinrel\relbar\joinrel\relbar
\joinrel\relbar\joinrel\relbar\joinrel\relbar} y;\\
\label{mp-d-groups-6} (a\circ_{\lambda} b)\overset{\phi_G(\lambda,a\circ_{\lambda}b)}{\leftharpoonup\joinrel\relbar
\joinrel\relbar\joinrel\relbar\joinrel\relbar} x&=& (a \overset{\phi_G(\lambda,a)}{\leftharpoonup\joinrel\relbar\joinrel\relbar\joinrel\relbar} (b\overset{\phi_G(\lambda,a)}{\relbar\joinrel\relbar\joinrel\rightharpoonup} x ))\circ_{\phi_H(\lambda,(a\circ_\lambda b)\overset{\lambda}{\rightharpoonup} x)} (b\overset{\phi_G(\lambda,a\circ_{\lambda}b)}{\leftharpoonup\joinrel\relbar\joinrel\relbar\joinrel\relbar} x),
\end{eqnarray}
for all $\lambda\in\lambda, a,b\in G,x,y\in H.$
\end{defi}

\begin{rmk}\label{mp-bijective}
By the fact that $\sigma: \mathscr{G} \otimes \mathscr{H}\to \mathscr{H} \otimes \mathscr{G}$ is a morphism of dynamical sets,   it follows that $\sigma$ satisfies the following equation:
\begin{eqnarray}\label{mp-d-groups-7}
\phi_G(\phi_H(\lambda,a\overset{\lambda}{\rightharpoonup} x),a\overset{\phi_G(\lambda,a)}{\leftharpoonup\joinrel\relbar\joinrel\relbar\joinrel\relbar} x)&=& \phi_H(\phi_G(\lambda,a),x).
\end{eqnarray}
\end{rmk}

\begin{pro}\label{mp-double}
Let $(\mathscr{G},\circ)$ and $(\mathscr{H},\cdot)$ be two dynamical groups. Then $(\mathscr{G},\mathscr{H},\sigma)$ is a matched pair of dynamical groups if and only if $(H\times G,\Lambda,\phi_{H\times G},\{\bowtie_\lambda\}_{\lambda\in \Lambda})$ is a dynamical group with the unit $(e_H,e_G)$, where the multiplication $\bowtie_\lambda$ and the structure map $\phi_{H\times G}:\Lambda \times H \times G\to \Lambda$ are given by
\begin{eqnarray}\label{double-d-groups}
(x,a) \bowtie_\lambda (y,b)=\Big(x \cdot_{\lambda} (a\overset{\phi_H(\lambda,x)}{\relbar\joinrel\relbar\joinrel\relbar\joinrel\rightharpoonup} y ),      (a\overset{\phi_G(\phi_H(\lambda,x),a)}{\leftharpoonup\joinrel\relbar\joinrel\relbar\joinrel
\relbar\joinrel\relbar\joinrel\relbar} y)\circ_{\phi_H(\lambda,x \cdot_{\lambda} (a\overset{\phi_H(\lambda,x)}{\relbar\joinrel\relbar\joinrel\relbar\joinrel\rightharpoonup} y )     )} b \Big),
\end{eqnarray}
and
\begin{eqnarray}
\phi_{H\times G}(\lambda,x,a)=\phi_G(\phi_H(\lambda,x),a),\quad \lambda\in\lambda, a,b\in G,x,y\in H.
\end{eqnarray}
Then the dynamical group $(H\times G,\Lambda,\phi_{H\times G},\{\bowtie_\lambda\}_{\lambda\in \Lambda})$ is called the {\bf double} of a matched pair of dynamical groups $(\mathscr{G},\mathscr{H},\sigma)$, and denoted by $\mathscr{H}\bowtie \mathscr{G}$.
\end{pro}

\begin{proof}
By a direct and elaborate calculation, $(H\times G,\Lambda,\phi_{H\times G},\{\bowtie_\lambda\}_{\lambda\in \Lambda})$ satisfies the dynamical associative law for all $\lambda\in\lambda, a,b,c\in G,x,y,z\in H,$
$$ ((x,a)\bowtie_\lambda (y,b))\bowtie_\lambda (z,c)=(x,a) \bowtie_\lambda ((y,b)\bowtie_{\phi_{H\times G}(\lambda,x,a)} (z,c)),   $$
if and only if Eqs.\, \eqref{mp-d-groups-2}-\eqref{mp-d-groups-3} and \eqref{mp-d-groups-5}-\eqref{mp-d-groups-7} hold. Moreover,
$(e_H,e_G) \bowtie_\lambda (x,a)=(x,a)$ if and only if Eq.\,\eqref{mp-d-groups-1} holds.
Similarly, $ (x,a) \bowtie_\lambda (e_H,e_G)=(x,a)$ if and only if Eq\, \eqref{mp-d-groups-4} holds. The inverse element $\overline{(x,a)}^{\lambda}$ of $(x,a)\in H\times G$ can be uniquely determined by \eqref{double-d-groups}.
\end{proof}

\begin{pro}\label{com-solution}
Let $(\mathscr{G},\circ)$ be a dynamical group. If a morphism of dynamical sets $R:\Lambda \times G \times G \to G \times G$ given by $R(\lambda,a,b)=(\varphi^{\lambda}_{a}(b),\psi^{\lambda}_{b}(a))$ satisfies the conditions
\begin{eqnarray}
\label{compatible-action-1} \varphi^{\lambda}_{a \circ_{\lambda} b}(c)&=&\varphi^{\lambda}_{a}~\varphi^{\phi_{G}(\lambda,a)}_{b}(c);\\
\label{compatible-action-2} \psi^{\lambda}_{b \circ_{\phi_{G}(\lambda,a)}  c}(a)&=&
\psi^{\phi_{G}(\lambda,\varphi^{\lambda}_{a}(b))}_{c}   \psi^{\lambda}_{b} (a) ;\\
\label{compatible-action-3} \varphi^{\lambda}_{a}(b)\circ_{\lambda}\psi^{\lambda}_{b}(a)&=&
a \circ_{\lambda} b,
\end{eqnarray}
for all $\lambda \in \Lambda, a,b,c \in G,$ then $R$ is a solution of the dynamical Yang-Baxter equation on the dynamical set $\mathscr{G}=(G,\Lambda,\phi_{G})$.
\end{pro}

\begin{proof}
%If $\varphi,\psi:\Lambda \times G \to G$ satisfy the Equations \eqref{compatible-action-1}-\eqref{compatible-action-3}, we now check that $R:\Lambda \times G \times G \to G \times G$ given by $R(\lambda,a,b)=(\varphi^{\lambda}_{a}(b),\psi^{\lambda}_{b}(a))$ is a solution of the dynamical Yang-Baxter equation \eqref{DYBE}.
The left-hand-side of \eqref{DYBE} is
\begin{eqnarray*}
R(\lambda)_{12} R(\phi_G(\lambda,G^{1}))_{23} R(\lambda)_{12}(a,b,c)
&=& R(\lambda)_{12} R(\phi_G(\lambda,G^{1}))_{23}(d,e,c)\\
&=& R(\lambda)_{12} (d,f,g)\\
&=& (h,k,g),
\end{eqnarray*}
where $ d=\varphi^{\lambda}_{a} (b),~ e=\psi^{\lambda}_{b} (a),~ f=\varphi^{\phi_{G}(\lambda,d)}_{e} (c),~ g=\psi^{\phi_{G}(\lambda,d)}_{c}(e),~ h=\varphi^{\lambda}_{d} (f)$ and $ k=\psi^{\lambda}_{f} (d)$.
By \eqref{compatible-action-3}, we have
$$ a \circ_{\lambda} b=d \circ_{\lambda} e,\quad e\circ_{\phi_{G}(\lambda,d)} c=f \circ_{\phi_{G}(\lambda,d)} g,\quad d \circ_{\lambda} f=h \circ_{\lambda} k.$$

The right-hand-side of \eqref{DYBE} is
\begin{eqnarray*}
R(\phi_G(\lambda,G^{1}))_{23}R(\lambda)_{12}R(\phi_G(\lambda,G^{1}))_{23}(a,b,c)
&=&R(\phi_G(\lambda,G^{1}))_{23}R(\lambda)_{12}(a,q,r)\\
&=&R(\phi_G(\lambda,G^{1}))_{23}(s,t,r)\\
&=&(s,v,w),
\end{eqnarray*}
where $ q=\varphi^{\phi_{G}(\lambda,a)}_{b} (c),~ r=\psi^{\phi_{G}(\lambda,a)}_{c}(b),~ s=\varphi^{\lambda}_{a} (q),~ t=\psi^{\lambda}_{q} (a),~ v=\varphi^{\phi_{G}(\lambda,s)}_{t} (r)$ and $ w=\varphi^{\phi_{G}(\lambda,s)}_{r} (t)$.
By \eqref{compatible-action-3}, we have
$$ b \circ_{\phi_{G}(\lambda,a)} c=q \circ_{\phi_{G}(\lambda,a)} r,\quad a \circ_{\lambda} q=s \circ_{\lambda} t, \quad t \circ_{\phi_{G}(\lambda,s)} r=v \circ_{\phi_{G}(\lambda,s)} w. $$
We use \eqref{compatible-action-1} to show that $h=s$ as follows:
\begin{eqnarray*}
s=\varphi^{\lambda}_{a} (q)=\varphi^{\lambda}_{a} (\varphi^{\phi_{G}(\lambda,a)}_{b} (c))=\varphi^{\lambda}_{a\circ_{\lambda} b}(c)
=\varphi^{\lambda}_{d\circ_{\lambda} e}(c)=\varphi^{\lambda}_{d} (\varphi^{\phi_{G}(\lambda,d)}_{e} (c))=\varphi^{\lambda}_{d} (f)=h.
\end{eqnarray*}
We use \eqref{compatible-action-2} to show that $w=g$ as follows:
\begin{eqnarray*}
w&=&\psi^{\phi_{G}(\lambda,s)}_{r} (t)=\psi^{\phi_{G}(\lambda,s)}_{r} (\varphi^{\lambda}_{q} (a))=\psi^{\phi_{G}(\lambda,\varphi_{a}^{\lambda}(q))}_{r}(\psi^{\lambda}_{q}(a))
=\psi^{\lambda}_{q\circ_{\phi_G(\lambda,a)} r}(a)\\
&=&\psi^{\lambda}_{b\circ_{\phi_{G}(\lambda,a)} c}(a)=\psi^{\phi_{G}(\lambda,\varphi_{a}^{\lambda}(b))}_{c}(\psi^{\lambda}_{b}(a))=
\psi^{\phi_{G}(\lambda,d)}_{c} (\varphi^{\lambda}_{b} (a))
=\psi^{\phi_{G}(\lambda,d)}_{c} (e)\\
&=&g.
\end{eqnarray*}
Finally, we show that $k=v$ by \eqref{compatible-action-3}. By a direct calculation, we have
\begin{eqnarray*}
(s \circ_{\lambda} v) \circ_{\lambda} w&=&s \circ_{\lambda} (v \circ_{\phi_G(\lambda,s)} w)=
s \circ_{\lambda} \Big(\varphi^{\phi_{G}(\lambda,s)}_{t} (r) \circ_{\phi_G(\lambda,s)}  \varphi^{\phi_{G}(\lambda,s)}_{r} (t) \Big)=s \circ_{\lambda} (t \circ_{\phi_G(\lambda,s)} r)\\
&=&(s \circ_{\lambda} t)\circ_{\lambda} r=(\varphi^{\lambda}_{a} (q) \circ_{\lambda} \psi^{\lambda}_{q} (a))\circ_{\lambda} r=(a \circ_{\lambda} q)\circ_{\lambda} r=a \circ_{\lambda}
(q \circ_{\phi_{G}(\lambda,a)} r)\\
&=&a \circ_{\lambda} (\varphi^{\phi_{G}(\lambda,a)}_{b} (c) \circ_{\phi_G(\lambda,a)} \psi^{\phi_{G}(\lambda,a)}_{c}(b))=a \circ_{\lambda} ( b \circ_{\phi_G(\lambda,a)} c)\\
&=&(a \circ_{\lambda} b)\circ_{\lambda} c,
\end{eqnarray*}
and
\begin{eqnarray*}
(h \circ_{\lambda} k) \circ_{\lambda} g&=&\Big(\varphi^{\lambda}_{d} (f) \circ_{\lambda} \psi^{\lambda}_{f} (d)\Big) \circ_{\lambda} g=(d \circ_{\lambda} f)\circ_{\lambda} g=d \circ_{\lambda} (f \circ_{\phi_G(\lambda,d)} g)\\
&=&d\circ_{\lambda} \Big( \varphi^{\phi_{G}(\lambda,d)}_{e} (c) \circ_{\phi_G(\lambda,d)} \psi^{\phi_{G}(\lambda,d)}_{c}(e) \Big)=d \circ_{\lambda} (e \circ_{\phi_{G}(\lambda,d)} c)
=(d \circ_{\lambda} e) \circ_{\lambda} c\\
&=&\Big(\varphi^{\lambda}_{a} (b) \circ_{\lambda} \psi^{\lambda}_{b} (a) \Big) \circ_{\lambda} c
\\&=&(a \circ_{\lambda} b)\circ_{\lambda} c,
\end{eqnarray*}
which implies that $(h \circ_{\lambda} k) \circ_{\lambda} g=(s \circ_{\lambda} v) \circ_{\lambda} w.$ Since $h=s$ and $g=w$, we can deduce that $k=v$ according to the cancellation law of the dynamical group $(G, \Lambda,\phi_{G},\{\circ_\lambda\}_{\lambda\in \Lambda})$. Therefore,    $R$ is a solution of the dynamical Yang-Baxter equation on the dynamical set $\mathscr{G}=(G,\Lambda,\phi_{G})$.
\end{proof}

\begin{rmk}
Eqs.\, \eqref{compatible-action-1}-\eqref{compatible-action-3} can be viewed as a generalization of compatible actions on groups given in \cite[Theorem 1]{LYZ}.
\end{rmk}

Next we introduce the notion of braided dynamical groups and show that braided dynamical groups give rise to solutions of the dynamical Yang-Baxter equation.

\begin{defi}\label{def-braid-d-group}
A {\bf braided dynamical group} is a pair $(\mathscr{G},\sigma)$, where $(\mathscr{G},\circ)$ is a dynamical group and $\sigma:\Lambda\times G \times G\to G \times G$ is a morphism of dynamical sets from $\mathscr{G} \otimes \mathscr{G}$ to $\mathscr{G} \otimes \mathscr{G}$ such that
\begin{itemize}
\item[{\rm(i)}] The triple $(\mathscr{G},\mathscr{G},\sigma)$ is a matched pair of dynamical groups.
\item[{\rm(ii)}] The following compatibility condition is satisfied:
\begin{eqnarray}
\label{braided-com} (a\overset{\lambda}{\rightharpoonup} b) \circ_{\lambda} (a\overset{\phi_G(\lambda,a)}{\leftharpoonup\joinrel\relbar\joinrel\relbar\joinrel\relbar} b)=
a \circ_{\lambda} b,\quad \forall \lambda\in\Lambda, a,b\in G.
\end{eqnarray}
\end{itemize}

For braided dynamical groups $(\mathscr{G},\sigma)$ and $(\mathscr{H},\sigma')$, a {\bf homomorphism of braided  dynamical groups} from $(\mathscr{G},\sigma)$ to $(\mathscr{H},\sigma')$ is a homomorphism of dynamical groups $\Psi: \mathscr{G} \to \mathscr{H}$ such that
\begin{eqnarray}\label{homo-braided-d-groups}
(\Psi\times \Psi)_{\lambda}\circ \sigma_{\lambda}=\sigma'_{\lambda} \circ (\Psi\times \Psi)_{\lambda}, \quad \forall \lambda\in\Lambda,
\end{eqnarray}
where the composition is given by \eqref{composition-maps}.
\end{defi}
Braided  dynamical groups and homomorphisms between braided  dynamical groups form a category, which is denoted by {\bf BDG}.

\begin{thm}\label{braided-solution}
Let $(\mathscr{G},\sigma)$ be a braided dynamical group. Then $\sigma$ is a non-degenerate solution of the dynamical Yang-Baxter equation on the dynamical set $\GGG=(G,\Lambda,\phi_G)$.
\end{thm}

\begin{proof}
Define $\varphi,\psi:\Lambda \times G \to G$ by
$$  \varphi^{\lambda}_{a}(b)=a\overset{\lambda}{\rightharpoonup} b,\quad
\psi^{\lambda}_{b}(a)=a\overset{\phi_G(\lambda,a)}{\leftharpoonup\joinrel\relbar\joinrel\relbar
\joinrel\relbar} b, \quad \forall \lambda\in \Lambda, a,b\in G.$$
Consider the conditions \eqref{mp-d-groups-2} and \eqref{mp-d-groups-5} of the matched pair of dynamical groups $(\mathscr{G},\mathscr{G},\sigma)$, we have
\begin{eqnarray*}
\varphi^{\lambda}_{a \circ_{\lambda} b}(c)
&=&(a\circ_\lambda b)\overset{\lambda}{\rightharpoonup} c
= a \overset{\lambda}{\rightharpoonup} (b\overset{\phi_G(\lambda,a)}{\relbar\joinrel\relbar\joinrel\rightharpoonup} c  )
=\varphi^{\lambda}_{a}~\varphi^{\phi_{G}(\lambda,a)}_{b}(c),\\
\varphi^{\lambda}_{a}(b)\circ_{\lambda}\psi^{\lambda}_{b}(a)&=&(a\overset{\lambda}{\rightharpoonup} b)\circ_{\lambda}(a\overset{\phi_G(\lambda,a)}
{\leftharpoonup\joinrel\relbar\joinrel\relbar\joinrel\relbar} b)\overset{\eqref{braided-com}}{=}
a \circ_{\lambda} b,
\end{eqnarray*}
and
\begin{eqnarray*}
\psi^{\lambda}_{b \circ_{\phi_{G}(\lambda,a)}  c}(a)&=&
a\overset{\phi_G(\lambda,a)}{\leftharpoonup\joinrel\relbar\joinrel\relbar\joinrel\relbar} (b \circ_{\phi_G(\lambda,a)} c)=(a\overset{\phi_G(\lambda,a)}{\leftharpoonup\joinrel\relbar\joinrel\relbar\joinrel\relbar} b)\overset{\phi_G(\phi_G(\lambda,a),b)}{\leftharpoonup\joinrel\relbar\joinrel\relbar\joinrel\relbar
\joinrel\relbar\joinrel\relbar\joinrel\relbar} c
\\&=&(a\overset{\phi_G(\lambda,a)}{\leftharpoonup\joinrel\relbar\joinrel\relbar\joinrel\relbar} b)\overset{\phi_G(\lambda,a \circ_{\lambda}b)}
{\leftharpoonup\joinrel\relbar\joinrel\relbar\joinrel\relbar
\joinrel\relbar\joinrel\relbar\joinrel\relbar} c \\
&=&(a\overset{\phi_G(\lambda,a)}{\leftharpoonup\joinrel\relbar\joinrel\relbar\joinrel\relbar} b)\overset{\phi_G(\lambda,\varphi^{\lambda}_{a}(b)\circ_{\lambda}\psi^{\lambda}_{b}(a))}
{\leftharpoonup\joinrel\relbar\joinrel\relbar\joinrel\relbar
\joinrel\relbar\joinrel\relbar\joinrel\relbar} c
\\&=&\psi^{\phi_{G}(\lambda,\varphi^{\lambda}_{a}(b))}_{c}   \psi^{\lambda}_{b} (a),
\end{eqnarray*}
which implies that $\varphi,\psi$ satisfy  \eqref{compatible-action-1}-\eqref{compatible-action-3}. Thus, by Proposition \ref{com-solution}, we deduce that $\sigma$ is a non-degenerate solution of the dynamical Yang-Baxter equation on the dynamical set $(G,\Lambda,\phi_G)$.
\end{proof}

\subsection{Braided dynamical groups and braided groupoids}

In this subsection, we first recall the embedding functor $Q$ from the category of dynamical sets to the category of quivers. Then we prove that this embedding functor $Q$ sends dynamical groups to groupoids. Moreover, we establish the relation between matched pairs of dynamical groups and matched pairs of groupoids using this embedding functor $Q$. As a special case, we show that braided dynamical groups can give rise to braided groupoids, which are associated to the quiver-theoretical solutions of the Yang-Baxter equation.

Let $\Lambda$ be a non-empty set. We denote by ${\bf Quiv_{\Lambda}}$ the category of quivers over $\Lambda$. Now let us recall from \cite{MS} the embedding functor $Q$ from ${\bf DSet_{\Lambda}}$ to ${\bf Quiv_{\Lambda}}$. Given an object $\mathscr{X}=(X,\Lambda,\phi_X) \in {\bf DSet_{\Lambda}}$, we can define the quiver $Q(\mathscr{X})=\Lambda \times X$, where the source and the target maps are  defined by:
$$ \alpha(\lambda,x)=\lambda, \quad \beta(\lambda,x)=\phi_X(\lambda,x),\quad \forall \lambda\in \Lambda,x\in X.$$
For any morphism $f:\mathscr{X}\to \mathscr{Y}$ in ${\bf DSet_{\Lambda}}$, we define
$$ Q(f):Q(\mathscr{X}) \to Q(\mathscr{Y}),\quad Q(f)(\lambda,x)=(\lambda,f(\lambda,x)),\quad \forall \lambda\in \Lambda,x\in X. $$
Then we obtain a functor $Q:{\bf DSet_{\Lambda}} \to {\bf Quiv_{\Lambda}}$. Moreover, the functor $Q$ is actually a fully faithful monoidal functor, which implies that $Q$ embeds the monoidal category ${\bf DSet_{\Lambda}}$ into ${\bf Quiv_{\Lambda}}$. See \cite[Theorem 2.7]{MS} for more details.

Now we show that this embedding functor $Q$ sends a dynamical group to a groupoid. Recall that a groupoid \cite{Mackenzie1} is a small category such that every morphism is invertible.

\begin{defi}
A {\bf groupoid} is a pair $(\G,M)$, where $M$ is the set of objects and $\G$ is the set of morphisms, with the  structure maps
\begin{itemize}
\item two surjective maps $\alpha,\beta: \G\longrightarrow M$, called the source map and target map, respectively;
\item  the multiplication $\star:\G_\beta\times_\alpha\G\longrightarrow \G$, where $\G_\beta\times_\alpha\G=\{(\gamma_1,\gamma_2)\in \G\times \G| \beta(\gamma_1)=\alpha(\gamma_2)\}$, such that
    $\alpha(\gamma_1\star \gamma_2)=\alpha(\gamma_1),~\beta(\gamma_1\star \gamma_2)=\beta(\gamma_2)$, for all $(\gamma_1,\gamma_2)\in\G_\beta\times_\alpha\G;$
\item  the inverse map $\inv:\G\longrightarrow \G$;
\item the inclusion map $\iota: M\longrightarrow \G$, $m\longmapsto \iota_m$ called the identity map;
\end{itemize}
satisfying the following properties:
\begin{enumerate}
\item[{\rm(i)}] {\rm(associativity)} $(\gamma_1\star \gamma_2)\star \gamma_3=\gamma_1\star (\gamma_2\star \gamma_3)$, whenever the multiplications are well-defined;
\item[{\rm(ii)}] {\rm(unitality)} $\iota_{\alpha(\gamma)}\star \gamma=\gamma=\gamma\star \iota_{\beta(\gamma)}$;
\item[{\rm(iii)}]  {\rm(invertibility)} $\gamma\star \inv(\gamma)=\iota_{\alpha(\gamma)}$, $\inv(\gamma)\star \gamma=\iota_{\beta(\gamma)}$.
\end{enumerate}
We also denote a groupoid by $\xymatrix{ (\G \ar@<0.5ex>[r]^{\alpha} \ar[r]_{\beta} & M},\star,\iota,\inv)$ or simply by $\G$. In particular, a {\bf bundle of groups} is defined to be a groupoid in which the source map and the target map are the same; while a {\bf group bundle} is a bundle of groups $\HH\stackrel{\pi}{\longrightarrow}M$ such that {$\HH_m:=\pi^{-1}(m)$} are isomorphic for all $m\in M$, and denoted by $(\HH\stackrel{\pi}{\longrightarrow}M,\star,\iota,\inv)$.
\end{defi}

\begin{pro}
Let $(\mathscr{G},\circ)$ be a dynamical group. Then
$$Q(\mathscr{G})=\xymatrix{ (\Lambda \times G \ar@<0.5ex>[r]^{\quad \alpha=\pr_1} \ar[r]_{\quad\beta=\phi_G} & \Lambda},\star,\iota,\inv)$$ is a groupoid, where the multiplication  $\star$ is given by
$$(\lambda,a)\star (\mu,b)=(\lambda, a\circ_{\lambda} b),\quad \forall \lambda\in \Lambda,a,b\in G~\mbox{satisfying}~\mu=\phi_{G}(\lambda, a),$$
the unit map $\iota $ and the inverse map $\inv $ are respectively given by
$$ \iota(\lambda)=(\lambda,e_G),\quad \inv(\lambda,a)=(\lambda,\bar{a}^\lambda), \quad \forall \lambda\in \Lambda, a\in G.$$

\end{pro}

\begin{proof}
By a direct calculation, the associativity of $\star$ holds according to Condition $\rm{(i)}$ in Definition \ref{dynamical-group}. Similarly, the unitality and invertibility hold by Conditions $\rm{(ii)}$ and $\rm{(iii)}$ in Definition \ref{dynamical-group} respectively. Thus, $Q(\mathscr{G})=\xymatrix{ (\Lambda \times G \ar@<0.5ex>[r]^{\quad \alpha=\pr_1} \ar[r]_{\quad\beta=\phi_G} & \Lambda},\star,\iota,\inv)$ is a groupoid.
\end{proof}

\begin{cor}
Let $(\mathscr{H},\cdot)$ be a constant dynamical group. Then $Q(\mathscr{H})=(\Lambda \times H \stackrel{\pr_1}{\longrightarrow} \Lambda,\star,\iota,\inv)$ is a bundle of groups, where $\star$ is given by
$$(\lambda,x)\star (\lambda,y)=(\lambda, x\cdot_{\lambda} y),\quad\forall \lambda\in \Lambda,x,y\in H.$$

 In addition, if $(\mathscr{H},\cdot)$ is a trivial dynamical group, then $Q(\mathscr{H})=(\Lambda \times H \stackrel{\pr_1}{\longrightarrow} \Lambda,\star,\iota,\inv)$ is a group bundle.
\end{cor}

\begin{proof}
 Since $\phi_H(\lambda,x)=\lambda$ for all $\lambda\in \Lambda$ in a constant dynamical group, it follows that $Q(\mathscr{H})_{\lambda}:=\pr_1^{-1}(\lambda)=(H,\cdot_{\lambda})$ are all groups, which implies that $Q(\mathscr{H})$ is a bundle of groups. Especially, if $(H,\cdot_{\lambda})$ are isomorphic for all $\lambda\in \Lambda$, then $Q(\mathscr{H})$ is a group bundle.
\end{proof}

\begin{ex}{\rm
Let $(\mathscr{G},\circ)=(\{0,1,2\}, \Lambda=\{\lambda_1,\lambda_2,\lambda_3\},\phi,\{\circ_{\lambda_i}\}_{\lambda_i\in \Lambda})$ be the dynamical group given in Example \ref{ex-d-group}. Then $Q(\mathscr{G})$ is a groupoid illustrated in the following figure:
\begin{figure}[htbp]
\centering
\begin{tikzpicture}[
    node distance=2cm,
    object/.style={
        circle,
        draw,
        fill=white,
        inner sep=1.5pt,
        minimum size=6mm
    },
    arrow/.style={
        ->,
        >=triangle 45,
        shorten >=1pt,
        line width=0.6pt
    },
    loop/.style={
        out=60,
        in=120,
        looseness=8,
        min distance=3mm
    },
    loop2/.style={
        out=150,
        in=210,
        looseness=9,
        min distance=4mm
    },
    loop3/.style={
        out=330,
        in=30,
        looseness=9,
        min distance=4mm
    },
    label/.style={
        midway,
        fill=white,
        inner sep=0.8pt,
        font=\small
    }
]

% 定义三个对象
\node[object] (l1) at (90:2) {$\scriptstyle\lambda_1$};
\node[object] (l2) at (210:2) {$\scriptstyle\lambda_2$};
\node[object] (l3) at (330:2) {$\scriptstyle\lambda_3$};

% 自同构箭头 - 使用不同的样式避免交叉
\draw[arrow] (l1) to [loop, above] node[label] {${\rm (\lambda_1,0)}$} (l1);
\draw[arrow] (l2) to [loop2, left] node[label, xshift=-2pt] {${\rm (\lambda_2,0)}$} (l2);
\draw[arrow] (l3) to [loop3, right] node[label, xshift=2pt] {${\rm (\lambda_3,0)}$} (l3);

% λ1 和 λ2 之间的互逆箭头
\draw[arrow] (l1) to [bend right=20] node[label, below left] {${\rm (\lambda_1,2)}$} (l2);
\draw[arrow] (l2) to [bend right=-5] node[label, above right] {${\rm (\lambda_2,2)} $} (l1);

% λ2 和 λ3 之间的互逆箭头
\draw[arrow] (l2) to [bend right=15] node[label, below right] {${\rm (\lambda_3,2)}$} (l3);
\draw[arrow] (l3) to [bend right=15] node[label, above left] {${\rm (\lambda_2,1)}$} (l2);

% λ3 和 λ1 之间的互逆箭头
\draw[arrow] (l3) to [bend right=20] node[label, below] {${\rm \quad \quad \quad (\lambda_1,1)}$} (l1);
\draw[arrow] (l1) to [bend right=15] node[label, above] {${\rm (\lambda_3,1)}$} (l3);

% 群胚关系说明
\node[align=center, font=\footnotesize] at (0,-2.3) {
    \scriptsize a groupoid with three objects $\lambda_1,\lambda_2,\lambda_3$ \\
    \scriptsize and nine arrows induced by $G$ and $\phi$ };
\end{tikzpicture}
\end{figure}
}\end{ex}

Let us recall the notion of matched pairs of groupoids.

\begin{defi}\cite{Mackenzie}
A {\bf matched pair of groupoids} is a triple $(\G,\HH,\sigma)$, where $\xymatrix{ \G \ar@<0.5ex>[r]^{\alpha} \ar[r]_{\beta} & M}$ and $\xymatrix{ \HH \ar@<0.5ex>[r]^{s} \ar[r]_{t} & M}$ are groupoids over the same base $M$ and
$$\sigma:\G_{\beta}\times_{s} \HH\lon \HH_{t}\times_{\alpha} \G,\,\,\,\,(\gamma,\delta)\mapsto(\gamma\rightharpoonup \delta,\gamma\leftharpoonup \delta),\,\,\,\forall(\gamma,\delta)\in\G_{\beta}\times_{s} \HH,$$
is a homomorphism of quivers satisfying the following conditions:
\begin{eqnarray}
\label{MG-1}\iota_{s(\delta)}\rightharpoonup \delta&=&\delta,\\
\label{MG-2}\gamma_1\rightharpoonup(\gamma_2\rightharpoonup \delta)&=&(\gamma_1\star_\G\gamma_2)\rightharpoonup \delta,\\
\label{MG-3}(\gamma_1\star_\G\gamma_2)\leftharpoonup \delta&=&\big(\gamma_1\leftharpoonup(\gamma_2\rightharpoonup \delta)\big)\star_\G(\gamma_2\leftharpoonup \delta),\\
\label{MG-4}\gamma\leftharpoonup \iota_{\beta(\gamma)}&=& \gamma,\\
\label{MG-5}(\gamma\leftharpoonup \delta_1)\leftharpoonup \delta_2&=&\gamma\leftharpoonup(\delta_1\star_\HH\delta_2),\\
\label{MG-6}\gamma\rightharpoonup(\delta_1\star_\HH\delta_2)&=&(\gamma\rightharpoonup \delta_1)\star_\HH\big((\gamma\leftharpoonup \delta_1)\rightharpoonup \delta_2\big),
\end{eqnarray}
for all $\gamma,\gamma_1,\gamma_2\in\G$ and $\delta,\delta_1,\delta_2\in\HH$ such that the compositions are possible.
\end{defi}

\begin{pro}\label{pro:vacant-groupoid}{\rm \cite{Andruskiewitsch,Mackenzie}}
Let $\xymatrix{ (\G \ar@<0.5ex>[r]^{\alpha} \ar[r]_{\beta} & M,}\star_\G,\iota^\G,\inv_\G)$ and $\xymatrix{ (\HH \ar@<0.5ex>[r]^{s} \ar[r]_{t} & M,}\star_\HH,\iota^\HH,\inv_\HH)$ be two groupoids. Then $(\G,\HH,\sigma)$ is a matched pair of groupoids if and only if $\xymatrix{(\HH_{t} \times_{\alpha} \G \ar@<0.5ex>[r]^{\quad s} \ar[r]_{\quad \beta} & M,} \bowtie,  \iota^{\bowtie},\inv_{\bowtie})$ is a groupoid, where the multiplication $\bowtie$ and  $\iota^{\bowtie},\inv_{\bowtie}$ are respectively given by
$$ (\delta_1,\gamma_1)\bowtie (\delta_2,\gamma_2)=\Big(\delta_1 \star_\HH (\gamma_1 \rightharpoonup \delta_2),(\gamma_1 \leftharpoonup \delta_2)\star_\G \gamma_2  \Big),$$
and
$$  \iota^{\bowtie}(m)=(\iota^\HH(m),\iota^\G(m)),\quad \inv_{\bowtie}(\delta_1,\gamma_1)=\Big(
\inv_\G \gamma_1 \rightharpoonup \inv_\HH \delta_1,\inv_\G ( \gamma_1 \leftharpoonup (\inv_\G \gamma_1 \rightharpoonup \inv_\HH \delta_1 )) \Big),  $$
for all $\gamma_1,\gamma_2\in\G$ and $\delta_1,\delta_2\in\HH$.
\end{pro}

The groupoid $\xymatrix{(\HH_{t} \times_{\alpha} \G \ar@<0.5ex>[r]^{\quad s} \ar[r]_{\quad \beta} & M}, \bowtie, \iota^{\bowtie},\inv_{\bowtie} )$ is called the {\bf vacant double groupoid}, which is denoted by $\HH_{t} \bowtie_{\alpha} \G$.

\begin{pro}\label{mp-groupoid-double}
Let $(\mathscr{G},\mathscr{H},\sigma)$ be a matched pair of dynamical groups. Assume that $Q(\mathscr{G})=\xymatrix{ (\Lambda \times G \ar@<0.5ex>[r]^{\quad \alpha=\pr_1} \ar[r]_{\quad\beta=\phi_G} & \Lambda},\star_{Q(\mathscr{G})},\iota,\inv)$ and $Q(\mathscr{H})=\xymatrix{ (\Lambda \times H \ar@<0.5ex>[r]^{\quad s=\pr_1} \ar[r]_{\quad t=\phi_H} & \Lambda},\star_{Q(\mathscr{H})},\iota,\inv)$ are the associated groupoids of $\mathscr{G}$ and $\mathscr{H}$ respectively. Then $(Q(\mathscr{G}),Q(\mathscr{H}),\hat{\sigma})$ is a
matched pair of groupoids, where $\hat{\sigma}: (\Lambda \times G)_{\beta}\times_{s} (\Lambda \times H) \lon (\Lambda \times H)_{t} \times_{\alpha} (\Lambda \times G) $ is given by
\begin{eqnarray}\label{mp-mp-groupoid}
\hat{\sigma}(\lambda,a,\lambda',x)=\Big(\lambda, a\overset{\lambda}{\rightharpoonup} x,\phi_H(\lambda,a\overset{\lambda}{\rightharpoonup} x),
a\overset{\phi_G(\lambda,a)}{\leftharpoonup\joinrel\relbar\joinrel\relbar\joinrel\relbar} x \Big),
\end{eqnarray}
for all $(\lambda,a)\in Q(\mathscr{G})$ and $(\lambda',x)\in Q(\mathscr{H})$ satisfying $\lambda'=\phi_G(\lambda,a)$.
\end{pro}

\begin{proof}
Let $(\mathscr{G},\mathscr{H},\sigma)$ be a matched pair of dynamical groups. By Proposition \ref{mp-double}, its double $\mathscr{H}\bowtie \mathscr{G}$ is a dynamical group. We use the embedding functor $Q$ acting on it and obtain the corresponding groupoid $Q(\mathscr{H}\bowtie \mathscr{G})$, which is obviously isomorphic to the vacant double groupoid $\xymatrix{Q(\mathscr{H})_{t} \bowtie_{\alpha} Q(\mathscr{G}) \ar@<0.5ex>[r]^{\quad\quad\quad s} \ar[r]_{\quad\quad\quad \beta} & \Lambda.}$ Then by Proposition \ref{pro:vacant-groupoid}, the latter is equivalent to the matched pair of groupoids $(Q(\mathscr{G}),Q(\mathscr{H}),\hat{\sigma})$, where $\hat{\sigma}$ is defined by \eqref{mp-mp-groupoid}.
\end{proof}

\begin{defi}{\rm (\cite{Andruskiewitsch,MM}) }
A {\bf braided groupoid}  is a pair $(\G,\sigma)$, where $\G$ is a groupoid and $\sigma:\G_\beta\times_\alpha \G\lon \G_\beta\times_\alpha \G$ is a homomorphism of quivers such that
\begin{enumerate}
\item[\rm(i)] the triple $(\G,\G,\sigma)$ is a matched pair of groupoids,
\item[\rm(ii)] $(\gamma\rightharpoonup \delta)\star(\gamma \leftharpoonup \delta)=\gamma\star\delta$, for all $(\gamma,\delta)\in \G_\beta\times_\alpha \G$.
\end{enumerate}
\end{defi}

\begin{thm}\label{braided-d-group-braided-gpd}
Let $(\mathscr{G},\sigma)$ be a braided dynamical group. Then $(Q(\mathscr{G}),\hat{\sigma})$ is a braided groupoid,
where $\hat{\sigma}: (\Lambda \times G)_{\beta}\times_{\alpha} (\Lambda \times G) \lon (\Lambda \times G)_{\beta} \times_{\alpha} (\Lambda \times G) $ is given by
\begin{eqnarray}\label{braided-braided-groupoid}
\hat{\sigma}(\lambda,a,\lambda',b)=(\lambda, a\overset{\lambda}{\rightharpoonup} b,\phi_G(\lambda,a\overset{\lambda}{\rightharpoonup} b),
a\overset{\phi_G(\lambda,a)}{\leftharpoonup\joinrel\relbar\joinrel\relbar\joinrel\relbar} b),
\end{eqnarray}
for all $(\lambda,a),(\lambda',b)\in Q(\mathscr{G})$ satisfying $\lambda'=\phi_G(\lambda,a)$.
\end{thm}

\begin{proof}
By Proposition \ref{mp-groupoid-double}, it follows that if $(\mathscr{G},\mathscr{G},\sigma)$ is a matched pair of dynamical groups, then
$(Q(\mathscr{G}),Q(\mathscr{G}),\hat{\sigma})$ is a matched pair of groupoids. Moreover, since $(\mathscr{G},\sigma)$ is a braided dynamical group, by \eqref{braided-com}, we have
\begin{eqnarray*}
(\lambda, a\overset{\lambda}{\rightharpoonup} b)\star
(\phi_G(\lambda,a\overset{\lambda}{\rightharpoonup} b),
a\overset{\phi_G(\lambda,a)}{\leftharpoonup\joinrel\relbar\joinrel\relbar\joinrel\relbar} b)
&=& (\lambda,(a\overset{\lambda}{\rightharpoonup} b)\circ_{\lambda} (a\overset{\phi_G(\lambda,a)}{\leftharpoonup\joinrel\relbar\joinrel\relbar\joinrel\relbar} b))\\
&=& (\lambda, a \circ_{\lambda} b)\\
&=& (\lambda, a)\star (\lambda',b),
\end{eqnarray*}
which implies that $(Q(\mathscr{G}),\hat{\sigma})$ is a braided groupoid.
\end{proof}

At the end of this section, we illustrate the role played by braided dynamical groups in the study of quiver-theoretical solutions of the Yang-Baxter equation.

The following concept was introduced by Andruskiewitsch in \cite{Andruskiewitsch}.

\begin{defi}{\rm (\cite{Andruskiewitsch}) }
Let $\xymatrix{\A \ar@<0.5ex>[r]^{\alpha} \ar[r]_{\beta} & \Lambda}$ be a quiver. A quiver-theoretical solution of the {\bf Yang-Baxter equation} on the quiver $\A$ is a quiver isomorphism $R:\A_{\beta}\times_{\alpha}\A\to \A_{\beta}\times_{\alpha}\A$ satisfying:
\begin{eqnarray}
R_{12}R_{23}R_{12}=R_{23}R_{12}R_{23}:\A_{\beta}\times_{\alpha}\A_{\beta}\times_{\alpha}\A\lon \A_{\beta}\times_{\alpha}\A_{\beta}\times_{\alpha}\A,
\end{eqnarray}
where $R_{12}=R\times\Id_\A,~R_{23}=\Id_\A\times R$. A quiver $\A$ with a quiver-theoretical solution of the Yang-Baxter equation on $\A$ is called a {\bf braided quiver}. Denote by $\Br({\bf Quiv_{\Lambda}})$ the category of  braided quivers over $\Lambda$.
\end{defi}

\begin{rmk}
A quiver-theoretical solution of the Yang-Baxter equation on a quiver over $\Lambda$ is a  solution of the Yang-Baxter equation in the monoidal category ${\bf Quiv_{\Lambda}}$.
\end{rmk}

\begin{thm}{\rm (\cite{Andruskiewitsch}) }\label{QYBE}
Let $(\G,\sigma)$ be a braided groupoid. Then $\sigma:\G_\beta\times_\alpha \G\lon \G_\beta\times_\alpha \G$ is a non-degenerate  quiver-theoretical solution of the  Yang-Baxter equation on the quiver $\xymatrix{\G \ar@<0.5ex>[r]^{\alpha} \ar[r]_{\beta} & \Lambda}$.
\end{thm}

On the other hand, Matsumoto and  Shimizu showed that a braided dynamical set gives rise to a braided quiver.

\begin{thm}{\rm (\cite{MS})}\label{br-dset-br-quiver}
The embedding functor $Q:{\bf DSet_{\Lambda}}\to{\bf Quiv_{\Lambda}}$ induces a fully faithful functor
$$ \Br(Q):\Br({\bf DSet_{\Lambda}}) \to \Br({\bf Quiv_{\Lambda}}), $$
which sends a braided object $(\mathscr{X},\sigma)$ of ${\bf DSet_{\Lambda}}$ to the braided object $(Q(\mathscr{X}),\widetilde{\sigma})$ of ${\bf Quiv_{\Lambda}}$, where
$$ \widetilde{\sigma}:=(Q_{X,X}^{(2)})^{-1} \circ Q(\sigma) \circ Q_{X,X}^{(2)}, $$
and
$$ Q_{X,X}^{(2)}: Q(\mathscr{X})_\beta \times_{\alpha} Q(\mathscr{X}) \to  Q(\mathscr{X} \otimes \mathscr{X}),\quad (\lambda,x,\mu,y) \mapsto (\lambda,x,y), $$
for all $\lambda\in \Lambda,x,y\in X$ and $\mu=\phi_X(\lambda,x)$.
\end{thm}

The following conclusion shows that braided dynamical groups provide a connection of the above two approaches to construct quiver-theoretical solutions of the Yang-Baxter equation.

\begin{cor} \label{c:twoconst}
Let $(\mathscr{G},\sigma)$ be a braided dynamical group. Then $\hat{\sigma}: (\Lambda \times G)_{\beta}\times_{\alpha} (\Lambda \times G) \lon (\Lambda \times G)_{\beta} \times_{\alpha} (\Lambda \times G) $ is a non-degenerate  quiver-theoretical solution of the  Yang-Baxter equation on the quiver $Q(\mathscr{G})=\xymatrix{ \Lambda \times G \ar@<0.5ex>[r]^{\quad \alpha=\pr_1} \ar[r]_{\quad\beta=\phi_G} & \Lambda}$, where $\hat{\sigma}$ is give by \eqref{braided-braided-groupoid}. Moreover, the following diagram is commutative:\begin{equation*}
\vcenter{\hbox{\footnotesize % 正确调整字体大小
\xymatrix@C=1.5em{  % 调整列间距
& *+[F]\txt{\footnotesize braided dynamical group}
  \ar^{ \text{Thm \ref{braided-solution}} }[d]
  \ar^{ \text{Thm \ref{braided-d-group-braided-gpd}} }[rr]
& & *+[F]\txt{\footnotesize braided groupoid}
  \ar^{ \text{Thm \ref{QYBE}} }[d] \\
& *+[F]\txt{\footnotesize solution of DYBE}
  \ar_{ \text{Thm \ref{br-dset-br-quiver}}\quad\quad\quad }[rr]
& & *+[F]\txt{\footnotesize quiver-theoretical solution of YBE}
}}}
\end{equation*}
\end{cor}

\begin{proof}
  The conclusion follows from Theorem \ref{QYBE}, Theorem  \ref{braided-d-group-braided-gpd}, Theorem \ref{braided-solution} and Theorem \ref{br-dset-br-quiver} directly.
\end{proof}

\vspace{-.4cm}
\section{Relative Rota-Baxter operators and matched pairs of dynamical groups}\label{sec:relative-RBO-mp}

In this section, we introduce the notion of a relative Rota-Baxter operator on a dynamical group, and show that relative Rota-Baxter operators on dynamical groups can give rise to matched pairs of dynamical groups  and thus solutions of the dynamical Yang-Baxter equation.

Let $\mathscr{G}=(G,\Lambda,\phi_\G)$ be a dynamical set and $\mathscr{H}=(H,\Lambda)$ a constant dynamical set. Denote by $\Map(H)$ the set of maps from $H$ to $H$, and by $\Perm(H)$ the set of bijective maps from $H$ to $H$.

\begin{defi}\label{weak-act}
Let $(\mathscr{G},\circ)$ be a dynamical semi-group and $(\mathscr{H},\cdot)$ be a constant dynamical group. A family of maps $\{\Phi_\lambda:G \to \Map(H)\}_{\lambda\in\Lambda}$ is called an {\bf action} of $\mathscr{G}$ on $\mathscr{H}$, if   the following conditions are satisfied:
\begin{eqnarray}
\label{action-1}\Phi_\lambda(a)(x \cdot_{\phi_{G}(\lambda,a)} y)&=&(\Phi_\lambda(a)x)\cdot_{\lambda} (\Phi_\lambda(a)y),\\
\label{action-2}\Phi_\lambda(a\circ_{\lambda} b)(x)&=& \Phi_\lambda(a)\Phi_{\phi_{G}(\lambda,a)}(b)(x),\quad \forall a,b\in G,x,y\in H.
\end{eqnarray}
\end{defi}

The notion of actions of   a dynamical group $(\mathscr{G},\circ)$ on a constant dynamical group $(\mathscr{H},\cdot)$ is naturally obtained from the above definition. Note that in this case, $\Phi_\lambda(a)$ must be invertible, i.e.  $\Phi_\lambda(a)\in\Perm(H)$.

\begin{pro}
Let  $\{\Phi_\lambda:G \to \Perm(H)\}_{\lambda\in\Lambda}$ be an action of a  dynamical group $(\mathscr{G},\circ)$ on a constant dynamical group $(\mathscr{H},\cdot)$. Then $(H\times G,\Lambda,\phi_{H\times G},\{\rtimes_\lambda\}_{\lambda\in \Lambda}) $ is a dynamical group, which is called the {\bf semi-direct product of $\mathscr{G}$ and $\mathscr{H}$} and simply denoted by $\mathscr{H} \rtimes_{\Phi} \mathscr{G}$, where the multiplication $\rtimes_\lambda$ is given by
\begin{eqnarray}\label{semi-direct-product}
(x,a)\rtimes_\lambda (y,b):=(x\cdot_\lambda (\Phi_\lambda(a)y),a \circ_\lambda b),\quad \forall \lambda\in \Lambda, a,b\in G,x,y\in H,
\end{eqnarray}
and the map $\phi_{H\times G}:\Lambda \times (H\times G)\to \Lambda$ is given by
\begin{eqnarray}\label{semi-direct-phi}
  \phi_{H\times G}(\lambda,x,a):=\phi_{G}(\lambda,a).
\end{eqnarray}
\end{pro}

\begin{proof}
By \eqref{semi-direct-product} and \eqref{semi-direct-phi}, for all $x,y,z\in H,a,b,c\in G$, we have
\begin{eqnarray*}
&&(x,a) \rtimes_{\lambda} ( (y,b) \rtimes_{\phi_{H\times G}(\lambda,x,a)} (z,c) )\\
&=& (x,a) \rtimes_{\lambda} \Big( (y\cdot_{\phi_G(\lambda,a)}(\Phi_{\phi_G(\lambda,a)}(b)z)),b\circ_{\phi_G(\lambda,a)} c \Big)\\
&=& \Big(x \cdot_{\lambda} \Phi_{\lambda}(a)(y \cdot_{\phi_G(\lambda,a)}(\Phi_{\phi_G(\lambda,a)}(b)z)),a \circ_{\lambda}  (b \circ_{\phi_G(\lambda,a)} c) \Big)\\
&=& \Big(  (x \cdot_{\lambda} (\Phi_{\lambda}(a)y))\cdot_{\lambda} (\Phi_{\lambda}(a \circ_{\lambda} b)z), (a\circ_{\lambda} b)  \circ_{\lambda} c   \Big)\\
&=& (x\cdot_\lambda (\Phi_\lambda(a)y),a \circ_\lambda b) \rtimes_{\lambda} (z,c)\\
&=& ((x,a)\rtimes_\lambda (y,b)) \rtimes_{\lambda} (z,c),
\end{eqnarray*}
which implies that $\{\rtimes_\lambda\}_{\lambda\in \Lambda}$ satisfies Condition {\rm(i)} in  Definition \ref{dynamical-group}. By a direct calculation, we can check that $(e_H,e_G)$ is the unit element satisfying $\phi_{H\times G}(\lambda,e_H,e_G)=\phi_{G}(\lambda,e_G)=\lambda$ and $ (e_H,e_G) \rtimes_{\lambda} (x,a)=(x,a) \rtimes_{\lambda} (e_H,e_G)=(x,a),$ for all $x\in H,a\in G$ and $\lambda\in \Lambda$. Moreover,
for all $\lambda\in \Lambda,(x,a)\in H\times G$, there exists an element $\Big(\Phi_{\phi_{G}(\lambda,a)}(\bar{a}^{\lambda})(x^{\lambda}),\bar{a}^{\lambda} \Big)\in H\times G$ such that
$$(x,a) \rtimes_{\lambda} \Big(\Phi_{\phi_{G}(\lambda,a)}(\bar{a}^{\lambda})(x^{\lambda}),\bar{a}^{\lambda} \Big)
=\Big(\Phi_{\phi_{G}(\lambda,a)}(\bar{a}^{\lambda})(x^{\lambda}),\bar{a}^{\lambda} \Big)
\rtimes_{\phi_{H\times G}(\lambda,x,a)}(x,a)=(e_H,e_G).$$
Thus, $(H\times  G,\Lambda,\phi_{H\times  G},\{\rtimes_\lambda\}_{\lambda\in \Lambda}) $ is a dynamical group.
\end{proof}

\begin{defi}
Let  $\{\Phi_\lambda:G \to \Perm(H)\}_{\lambda\in\Lambda}$ be an action of a  dynamical group $(\mathscr{G},\circ)$ on a constant dynamical group $(\mathscr{H},\cdot)$. A map $\huaB:H \to G$ is called a {\bf relative Rota-Baxter operator} on $\mathscr{G}$   if the following equality holds for all $\lambda\in \Lambda,x,y\in H$,
\begin{eqnarray}\label{rRBO}
\huaB(x) \circ_{\lambda} \huaB(y)=\huaB( x\cdot_{\lambda} \Phi_{\lambda}(\huaB (x)) (y)).
\end{eqnarray}
\end{defi}

\begin{lem}
  Let $\huaB:H \to G$ be a relative Rota-Baxter operator on a dynamical group $(\mathscr{G},\circ)$ with respect to the action $\{\Phi_\lambda:G \to \Perm(H)\}_{\lambda\in\Lambda}$. Then  $\huaB(e_H)=e_G.$
\end{lem}
\begin{proof}
  By \eqref{rRBO}, taking $x,y$ to be the identity element $e_H$, we have $\huaB(e_H) \circ_{\lambda} \huaB(e_H)=\huaB(e_H)$ for all $\lambda\in \Lambda$, which implies that $\huaB(e_H)=e_G.$
\end{proof}

\begin{pro}
Let $\{\Phi_\lambda:G \to \Perm(H)\}_{\lambda\in\Lambda}$ be an action of a  dynamical group $(\mathscr{G},\circ)$ on a constant dynamical group $(\mathscr{H},\cdot)$. Then a map $\huaB:H \to G$ is a relative Rota-Baxter operator on $\mathscr{G}$ with respect to the action $\{\Phi_\lambda:G \to \Perm(H)\}_{\lambda\in\Lambda}$ if and only if the graph ${\rm Gr} (\huaB)=\{(x,\huaB(x))|x\in H\}$ is a dynamical subgroup of $\mathscr{H}\rtimes_{\Phi} \mathscr{G}$.
\end{pro}

\begin{proof}
For all $x,y\in H$, we have
$(x,\huaB(x))\rtimes_{\lambda}(y,\huaB(y))=(x\cdot_{\lambda}\Phi_{\lambda}(\huaB (x)) (y) , \huaB(x) \circ_{\lambda} \huaB(y))  $, which implies that ${\rm Gr} (\huaB)=\{(x,\huaB(x))|x\in H\}$ is a dynamical subgroup of $\mathscr{H}\rtimes_{\Phi} \mathscr{G}$ if and only if
$$\huaB(x) \circ_{\lambda} \huaB(y)=\huaB( x\cdot_{\lambda} \Phi_{\lambda}(\huaB (x)) (y)).$$
This means that $\huaB:H \to G$ is a relative Rota-Baxter operator on $\mathscr{G}$ with respect to the action $\{\Phi_\lambda:G \to \Perm(H)\}_{\lambda\in\Lambda}$. Actually, $(e_H,B(e_H))=(e_H,e_G)$ is the unit of the graph ${\rm Gr} (\huaB)=\{(x,\huaB(x))|x\in H\}$ and for any element $(x,B(x))\in {\rm Gr} (\huaB)$, its inverse $\big(\Phi_{\phi_{G}(\lambda,\huaB(x))}(\overline{\huaB(x)}^{\lambda})(x^{\lambda}),\overline{\huaB(x)}^{\lambda})$
belongs to the graph ${\rm Gr} (\huaB)$ by a direct calculation.
\end{proof}

\begin{pro}\label{pro-descendant}
Let $\huaB:H \to G$ be a relative Rota-Baxter operator on a  dynamical group $(\mathscr{G},\circ)$ with respect to the action $\{\Phi_\lambda:G \to \Perm(H)\}_{\lambda\in\Lambda}$. Define the multiplication $\circ^{\huaB}_{\lambda}:H \times H \to H$ by
\begin{eqnarray}\label{descendant-product}
x \circ^{\huaB}_{\lambda} y:=x\cdot_{\lambda} \Phi_{\lambda}(\huaB (x)) (y) \quad \forall \lambda\in \Lambda, x\in H,
\end{eqnarray}
and the map $\phi_{\huaB}:\Lambda\times H\to H$ by
\begin{eqnarray}\label{descendant-map}
\phi_{\huaB}(\lambda,x):=\phi_G(\lambda,\huaB (x)).
\end{eqnarray}
Then $(H,\Lambda,\phi_{\huaB},\{\circ^{\huaB}_\lambda\}_{\lambda\in \Lambda})$ is a dynamical group, which is called the {\bf descendant dynamical group} of $\huaB$, and denoted by $\mathscr{H}_{\huaB}$.

Moreover, the map $\widetilde{\huaB}:\Lambda \times H \to G$ defined by
$$\widetilde{\huaB}(\lambda,x)=\huaB(x),\quad \forall \lambda\in \Lambda,x\in H,$$
is a homomorphism from the dynamical group $(H,\Lambda,\phi_{\huaB},\{\circ^{\huaB}_\lambda\}_{\lambda\in \Lambda})$ to $(\mathscr{G},\circ)$.
\end{pro}

\begin{proof}
Let $\lambda\in \Lambda, x,y,z\in H$. By \eqref{action-1}-\eqref{action-2} and  \eqref{descendant-product}-\eqref{descendant-map}, we have
\begin{eqnarray*}
(x \circ^{\huaB}_\lambda y)\circ^{\huaB}_\lambda z
&=& (x\cdot_{\lambda} \Phi_{\lambda}(\huaB (x)) (y)) \circ^{\huaB}_\lambda z\\
&=& (x\cdot_{\lambda} \Phi_{\lambda}(\huaB (x)) (y)) \cdot_{\lambda} \Phi_{\lambda}(\huaB (x\circ^{\huaB}_{\lambda} y)) (z)\\
&=& (x\cdot_{\lambda} \Phi_{\lambda}(\huaB (x)) (y)) \cdot_{\lambda} \Phi_{\lambda}((\huaB x)\circ_\lambda (\huaB y)) (z)\\
&=& x\cdot_{\lambda} \Phi_{\lambda}(\huaB (x)) (y) \cdot_{\lambda} \Phi_{\lambda}(\huaB (x)) \Phi_{\phi_G(\lambda,\huaB x)}(\huaB (y)) (z)\\
&=& x\cdot_{\lambda} \Phi_{\lambda}(\huaB (x)) (y \cdot_{\phi_G(\lambda,\huaB x)}   \Phi_{\phi_G(\lambda,\huaB x)}(\huaB y) z) \\
&=& x\circ^{\huaB}_{\lambda} (y \cdot_{\phi_G(\lambda,\huaB (x))} \Phi_{\phi_G(\lambda,\huaB(x))}(\huaB (y)) (z) )\\
&=& x \circ^{\huaB}_\lambda (y\circ^{\huaB}_{\phi_G(\lambda,\huaB (x))} z)\\
&=& x \circ^{\huaB}_\lambda (y\circ^{\huaB}_{\phi_{\huaB}(\lambda,x)} z).
\end{eqnarray*}
Moreover, by a direct calculation, $e_H$ is the unit, and $\Phi_{\phi_\huaB(\lambda, x)} (\overline{\huaB (x)}^\lambda)(x^\lambda)$ is the inverse of $x$. Thus, $(H,\Lambda,\phi_{\huaB},\{\circ^{\huaB}_\lambda\}_{\lambda\in \Lambda})$ is a dynamical group.

Obviously we have
\begin{eqnarray*}
 \phi_{G} (\lambda, \widetilde{\huaB}_{\lambda}(x))&=&\phi_{G} (\lambda, \huaB(x))=\phi_{ {\huaB}}(\lambda,x);\\
\widetilde{\huaB}_{\lambda} (x \circ^{\huaB}_\lambda y)&=&\huaB(x \circ^{\huaB}_\lambda y)= \huaB (x \cdot_{\lambda} \Phi_{\lambda}(\huaB(x))(y))=
\huaB(x) \circ_{\lambda} \huaB(y)=
\widetilde{\huaB}_{\lambda}(x) \circ_{\lambda} \widetilde{\huaB}_{\phi_{{\huaB}}(\lambda,x)}(y),
\end{eqnarray*}
for all $\lambda\in\Lambda$ and $x,y\in H$,  which implies that $\widetilde{\huaB}$ is a homomorphism of dynamical groups from $(H,\Lambda,\phi_{\huaB},\{\circ^{\huaB}_\lambda\}_{\lambda\in \Lambda})$ to $(\GGG,\{\circ_\lambda\}_{\lambda\in \Lambda})$.
\end{proof}

Let $(\mathscr{G},\circ)$  be a dynamical group and
$(\mathscr{H},\cdot)$ be a constant dynamical group. For a map $\huaB:H \to G$, the map
\begin{eqnarray}
\xi^{\lambda}_{\huaB} :H \times G \to H \times G,\quad  \xi^{\lambda}_{\huaB} (x,a)=(x, \huaB(x) \circ_{\lambda} a), \quad \forall x\in H,a\in G,
\end{eqnarray}
is invertible. Indeed, the inverse map $(\xi^{\lambda}_{\huaB})^{-1}: H \times G \to H \times G$ is given by
\begin{eqnarray}
(\xi^{\lambda}_{\huaB})^{-1} (x,a)=(x, \overline{\huaB(x)}^{\lambda} \circ_{\phi_{\huaB}(\lambda,x)} a), \quad \forall x\in H,a\in G.
\end{eqnarray}

Pulling back the dynamical group structure on $\mathscr{H} \rtimes_{\Phi} \mathscr{G}$, we obtain a dynamical group $(H \times G,\Lambda,\widetilde{\phi_{\huaB}},\{\ast^{\huaB}_\lambda\}_{\lambda\in\Lambda})$, where the product $\ast^{\huaB}_\lambda$ is given by
\begin{eqnarray*}\label{d-group-fac}
&&(x,a)\ast^{\huaB}_{\lambda} (y,b)\\
&=& (\xi^{\lambda}_{\huaB})^{-1} (\xi^{\lambda}_{\huaB}(x,a) \rtimes_{\lambda} \xi^{\lambda}_{\huaB}(y,b) )\\
&=&  (\xi^{\lambda}_{\huaB})^{-1}((x, \huaB(x) \circ_{\lambda} a)\rtimes_{\lambda} (y, \huaB(y) \circ_{\lambda} b))\\
&=& (\xi^{\lambda}_{\huaB})^{-1} ( x \cdot_{\lambda}  \Phi_{\lambda} (\huaB (x)\circ_{\lambda} a)(y),(\huaB(x)\circ_{\lambda} a)         \circ_{\lambda}  (\huaB(y)\circ_{\lambda} b) )\\
&=&\Big( x \cdot_{\lambda} \Phi_{\lambda} (\huaB (x)\circ_{\lambda} a)(y), \overline{\huaB(x \cdot_{\lambda}  \Phi_{\lambda} (\huaB (x)\circ_{\lambda} a)(y))}^{\lambda} \circ_{\phi_{\huaB}(\lambda, x \cdot_{\lambda} \Phi_{\lambda} (\huaB (x)\circ_{\lambda} a)(y))} ((\huaB(x)\circ_{\lambda} a)         \circ_{\lambda}  (\huaB(y)\circ_{\lambda} b)  )    \Big),
\end{eqnarray*}
 the structure map $\widetilde{\phi_{\huaB}}:\Lambda\times H \times G \to \Lambda$ is given by
\begin{eqnarray*}
\widetilde{\phi_{\huaB}}(\lambda,x,a)= \phi_{G}(\lambda,\huaB (x)\circ_{\lambda} a),
\end{eqnarray*}
the unit is given by
$$  e=\xi_{\huaB}^{-1}(e_H,e_G)=(e_H,\overline{\huaB(e_H)}^{\lambda}), $$
and the inverse of $(x,a)$ is given by
\begin{equation}\label{d-group-inverse}
(x,a)^{\lambda}=\Big( \Phi_{\lambda}(\huaB (x)\circ_{\lambda} a)^{-1}(x^\lambda),  \overline{\huaB(\Phi_{\lambda}(\huaB (x)\circ_{\lambda} a)^{-1}(x^\lambda)) }^{\lambda} \circ_{\phi_{\huaB}(\lambda,\Phi_{\lambda}(\huaB (x)\circ_{\lambda} a)^{-1}(x^\lambda))}       \overline{\huaB (x)\circ_{\lambda} a}^{\lambda}   \Big).
\end{equation}

\begin{thm}\label{rRBO-fac}
With above notations, $(H \times G,\Lambda,\widetilde{\phi_{\huaB}},\{\ast^{\huaB}_\lambda\}_{\lambda\in\Lambda})$ is a dynamical group factorization into dynamical subgroups $H \times \{e_G\}$ and $\{e_H\} \times G$ if and only if $\huaB:H \to G$ is a relative Rota-Baxter operator on the dynamical group $(\mathscr{G},\circ)$ with respect to the action $\{\Phi_\lambda:G \to \Perm(H)\}_{\lambda\in\Lambda}$.
\end{thm}

\begin{proof}
Since $\huaB(e_H)=e_G$, for all $(e_H,a),(e_H,b)\in \{e_H\}\times G$, we have
$$ (e_H,a)\ast^{\huaB}_{\lambda} (e_H,b)= (e_H,a\circ_\lambda b). $$
For all $(e_H,a)\in \{e_H\}\times G$, we have
$$  (e_H,a)^{\lambda}=(e_H,\overline{a}^{\lambda}).     $$
Therefore, we deduce that $\{e_H\} \times G$ is a dynamical subgroup of $(H \times G,\Lambda,\widetilde{\phi_{\huaB}},\{\ast^{\huaB}_\lambda\}_{\lambda\in\Lambda})$.

For all $x,y\in H$, we have
\begin{equation}\label{H-eG}
(x,e_G)\ast^{\huaB}_{\lambda} (y,e_G)
=\Big(x\cdot_{\lambda} \Phi_{\lambda}(\huaB (x))(y), \overline{\huaB(x\cdot_{\lambda} \Phi_{\lambda}(\huaB (x))(y))}^{\lambda} \circ_{\phi_{\huaB}(\lambda,x\cdot_{\lambda} \Phi_{\lambda}(\huaB (x))(y))} ((\huaB x)\circ_{\lambda}(\huaB y))      \Big).
\end{equation}
Thus, if $H \times \{e_G\}$ is a dynamical subgroup of $(H \times G,\Lambda,\widetilde{\phi_{\huaB}},\{\ast^{\huaB}_\lambda\}_{\lambda\in\Lambda})$, then we get
$$  \overline{\huaB(x\cdot_{\lambda} \Phi_{\lambda}(\huaB (x))(y))}^{\lambda} \circ_{\phi_{\huaB}(\lambda,x\cdot_{\lambda} \Phi_{\lambda}(\huaB (x))(y))} ((\huaB x)\circ_{\lambda}(\huaB y))=e_G, $$
which implies that $\huaB:H \to G$ is a relative Rota-Baxter operator on the dynamical group $(\mathscr{G},\circ)$ with respect to the action $\{\Phi_\lambda:G \to \Perm(H)\}_{\lambda\in\Lambda}$.

Conversely, let $\huaB:H \to G$ be a relative Rota-Baxter operator. We have $\huaB(e_H)=e_G$.
By \eqref{H-eG}, we obtain $(x,e_G)\ast^{\huaB}_{\lambda} (y,e_G)\in H\times \{e_G\}$. For $(h,e_G)\in H\times e_G$, by \eqref{d-group-inverse}, we have
\begin{eqnarray*}
(x,e_G)^{\lambda}&=&\Big( \Phi_{\lambda}(\huaB (x))^{-1}(x^\lambda),  \overline{\huaB(\Phi_{\lambda}(\huaB (x))^{-1}(x^\lambda)) }^{\lambda} \circ_{\phi_{\huaB}(\lambda,\Phi_{\lambda}(\huaB (x))^{-1}(x^\lambda))}  \overline{\huaB ( x)}^{\lambda}   \Big)\\
&=& \Big( \Phi_{\lambda}(\huaB (x))^{-1}(x^\lambda), e_G  \Big).
\end{eqnarray*}
The last equality follows from
$$ \huaB (x)\circ_{\lambda} \huaB(\Phi_{\lambda}(\huaB (x))^{-1}(x^\lambda))=\huaB(x\cdot_{\lambda} \Phi_{\lambda}(\huaB (x))\Phi_{\lambda}(\huaB (x))^{-1} (x^\lambda))=\huaB(x \cdot_{\lambda} x^\lambda)=\huaB(e_H)=e_G. $$
Thus,  $(x,e_G)^{\lambda} \in H \times \{e_G\}$ and $H \times \{e_G\}$ is a dynamical subgroup of $(H \times G,\Lambda,\widetilde{\phi_{\huaB}},\{\ast^{\huaB}_\lambda\}_{\lambda\in\Lambda})$.

For all $\lambda\in \Lambda,x\in H, a \in G$, we have
$$ (x,e_G)\ast^{\huaB}_{\lambda} (e_H,a)=(x,a). $$
It is obvious that $e=(e_H, e_G)$ is the unit of the dynamical group $(H \times G,\Lambda,\widetilde{\phi_{\huaB}},\{\ast^{\huaB}_\lambda\}_{\lambda\in\Lambda})$ and $H \times \{e_G\} \cap \{e_H\} \times G=\{(e_H,e_G)\} .$ Therefore, $(H \times G,\Lambda,\widetilde{\phi_{\huaB}},\{\ast^{\huaB}_\lambda\}_{\lambda\in\Lambda})$ is a dynamical group factorization into dynamical subgroups $H \times \{e_G\}$ and $\{e_H\} \times G$.
\end{proof}

Relative Rota-Baxter operators on dynamical groups naturally give rise to matched pairs of dynamical groups.

\begin{pro}\label{rRBO-mp}
Let $\huaB:H \to G$ be a relative Rota-Baxter operator on a  dynamical group $(\mathscr{G},\circ)$ with respect to the action $\{\Phi_\lambda:G \to \Perm(H)\}_{\lambda\in\Lambda}$. Define $\rightharpoonup:\Lambda \times G \times H \to H$ and
$\leftharpoonup:\Lambda \times G \times H \to G$ respectively by
\begin{eqnarray}
a\overset{\lambda}{\rightharpoonup} x&=&\Phi_{\lambda}(a)(x);\\
a\overset{\phi_G(\lambda,a)}{\leftharpoonup\joinrel\relbar\joinrel\relbar\joinrel\relbar} x
&=& \overline{\huaB(\Phi_{\lambda}(a)(x)) }^{\lambda} \circ_{\phi_{G}(\lambda,\huaB(\Phi_{\lambda}(a)(x))} (a\circ_{\lambda}   \huaB (x)),
\end{eqnarray}
for all $\lambda\in \Lambda$ and $a\in G,x\in H$. Then $(\mathscr{G},\mathscr{H}_{\huaB},\sigma)$ is a matched pair of dynamical groups, where $ \sigma(\lambda,a,x)=(a\overset{\lambda}{\rightharpoonup} x,a\overset{\phi_G(\lambda,a)}{\leftharpoonup\joinrel\relbar\joinrel\relbar\joinrel\relbar}x). $
\end{pro}

\begin{proof}
It follows from the fact that $(H \times G,\Lambda,\widetilde{\phi_{\huaB}},\{\ast^{\huaB}_\lambda\}_{\lambda\in\Lambda})$ is a dynamical group factorization into dynamical subgroups $H \times \{e_G\}$ and $\{e_H\} \times G$ by Theorem \ref{rRBO-fac}.
\end{proof}
\vspace{-.3cm}
\section{Dynamical post-groups}\label{sec:d-post-group}

In this section, we introduce the notion of dynamical post-groups. A relative Rota-Baxter operator on a dynamical group naturally induces a dynamical post-group structure. Conversely, a dynamical post-group gives rise to a relative Rota-Baxter operator on its sub-adjacent dynamical group. Finally, we show that the category of dynamical post-groups is isomorphic to the category of braided dynamical groups.

In this section, as before, $\mathscr{G}=(G,\Lambda,\phi_G)$ is a dynamical set.

\subsection{Dynamical post-groups and relative Rota-Baxter operators}
\begin{defi}\label{d-skew-brace}

A {\bf weak dynamical post-group} structure on a dynamical set $\mathscr{G}$ consists of two families of maps $\{\cdot_\lambda\}_{\lambda\in \Lambda}$ and $\{\vartriangleright_\lambda\}_{\lambda\in \Lambda}$ from $G\times G$ to $G$ such that
\begin{itemize}
\item[{\rm(i)}] $(G,\Lambda,\{\cdot_\lambda\}_{\lambda\in \Lambda})$ is a constant dynamical group;

\item[{\rm(ii)}] $\{\cdot_\lambda\}_{\lambda\in \Lambda}$ and $\{\vartriangleright_\lambda\}_{\lambda\in \Lambda}$ satisfy the following compatibility conditions:
  \begin{eqnarray}
  \label{d-post-group1} a\vartriangleright_\lambda (b \cdot_{\phi_{G}(\lambda,a)} c)&=&(a \vartriangleright_\lambda b)\cdot_\lambda (a \vartriangleright_\lambda c), \\
  \label{d-post-group2}(a\cdot_{\lambda } (a\vartriangleright_\lambda b))\vartriangleright_\lambda c&=&a \vartriangleright_\lambda (b \vartriangleright_{\phi_G(\lambda,a)} c), \quad \forall \lambda\in \Lambda, a,b,c\in G.
  \end{eqnarray}
\end{itemize}
We denote a weak dynamical post-group by $(\mathscr{G},\{\cdot_\lambda\}_{\lambda\in \Lambda},\{\vartriangleright_\lambda\}_{\lambda\in \Lambda})$.

A {\bf dynamical post-group} is a weak dynamical post-group in which $L^{\rhd_\lambda}_a:G\to G$ is invertible for  all $\lambda\in \Lambda$ and $a\in G$, where  $L^{\rhd_\lambda}_a:G\to G$ is defined by
$$
L^{\rhd_\lambda}_a (b)=a \rhd_\lambda b.
$$

  A {\bf dynamical pre-group} is a  dynamical post-group $(\GGG,  \{\cdot_\lambda\}_{\lambda\in \Lambda},\{\vartriangleright_\lambda\}_{\lambda\in \Lambda})$ in which the groups $(G,\cdot_\lambda)$ are abelian for all $\lambda\in \Lambda.$
\end{defi}

\begin{ex}\label{d-post-group-R}
{\rm
Let $(\mathbb{R},+,\cdot)$ be the field of real numbers. Then $(\mathbb{R}, \Lambda=\mathbb{R},\phi,\{\cdot_\lambda\}_{\lambda\in \mathbb{R}},\{\vartriangleright_\lambda\}_{\lambda\in \mathbb{R}})$ is
a dynamical post-group, where $\cdot_\lambda:=+$ for all $\lambda\in \mathbb{R}$ and $\phi:\mathbb{R} \times \mathbb{R} \to \mathbb{R}$, $\vartriangleright_\lambda:\mathbb{R} \times \mathbb{R} \to \mathbb{R}$ are given by
$$ a \vartriangleright_\lambda b:= (\lambda a+1)^2 \cdot b, \quad \phi(\lambda,a):= \lambda\cdot (\lambda a+1),\quad \forall \lambda,a,b\in \mathbb{R}. $$
}
\end{ex}

\begin{ex}\label{ex-d-post-group}
{\rm
Let $G=(\mathbb{Z}_3,+)$ be a cyclic group of order $3$. Let $\Lambda=\{\lambda_1,\lambda_2,\lambda_3\}$. Then $\mathscr{G}=(G,\Lambda,\phi)$ is a dynamical set, where $\phi:\Lambda \times G \to \Lambda$ is given by
\begin{center}
\begin{minipage}{0.23\textwidth}
\centering
\(
\begin{array}{c|ccc}
    \phi & 0 & 1 & 2 \\ \hline
    \lambda_1 & \lambda_1 & \lambda_3 & \lambda_2 \\
    \lambda_2 & \lambda_2 & \lambda_3 & \lambda_1 \\
    \lambda_3 & \lambda_3 & \lambda_1 & \lambda_2 \\
\end{array}
\)
\par\vspace{0.25ex} % 减小垂直间距
\end{minipage}%
 \end{center}
Then $(\mathscr{G},\{\cdot_{\lambda_i}:=+ \}_{\lambda_i \in \Lambda},\{\vartriangleright_{\lambda_i}\}_{\lambda_i \in \Lambda})$ is a dynamical post-group, where $\vartriangleright_{\lambda_i}:G \times G \to G$ is given by
\begin{center}
\begin{minipage}{0.23\textwidth}
\centering
\(
\begin{array}{c|ccc}
    \vartriangleright_{\lambda_1} & 0 & 1 & 2 \\ \hline
    0 & 0 & 1 & 2 \\
    1 & 0 & 2 & 1 \\
    2 & 0 & 2 & 1 \\
\end{array}
\)
\par\vspace{0.25ex} % 减小垂直间距
\end{minipage}%
\hspace{1em} % 减小水平间距
\begin{minipage}{0.23\textwidth}
\centering
\(
\begin{array}{c|ccc}
    \vartriangleright_{\lambda_2} & 0 & 1 & 2 \\ \hline
    0 & 0 & 1 & 2 \\
    1 & 0 & 1 & 2 \\
    2 & 0 & 2 & 1 \\
\end{array}
\)
\par\vspace{0.25ex} % 减小垂直间距
\end{minipage}%
\hspace{1em} % 减小水平间距
\begin{minipage}{0.23\textwidth}
\centering
\(
\begin{array}{c|ccc}
    \vartriangleright_{\lambda_3} & 0 & 1 & 2 \\ \hline
    0 & 0 & 1 & 2 \\
    1 & 0 & 2 & 1 \\
    2 & 0 & 1 & 2 \\
\end{array}
\)
\par\vspace{0.25ex} % 减小垂直间距
\end{minipage}
\end{center}
}\end{ex}

\begin{lem}
Let $e_G$ be the unit for $\{\cdot_\lambda\}_{\lambda\in \Lambda}$ in a dynamical post-group $(\GGG,\{\cdot_\lambda\}_{\lambda\in \Lambda},\{\vartriangleright_\lambda\}_{\lambda\in \Lambda})$. Then for all $\lambda\in \Lambda$ and $a\in G$, we have
\begin{eqnarray}
\label{unit-d-post-group1} a \vartriangleright_\lambda e_G&=&e_G,\\
\label{unit-d-post-group2} e_G \vartriangleright_\lambda a &=& a.
\end{eqnarray}
\end{lem}

\begin{proof}
Taking  $b=c=e_G$ in \eqref{d-post-group1}, we have
$a\vartriangleright_\lambda e_G=(a \vartriangleright_\lambda e_G)\cdot_\lambda (a \vartriangleright_\lambda e_G)$,
which implies that $a \vartriangleright_\lambda e_G=e_G$. Taking $a=b=e_G$ in \eqref{d-post-group2}, we have
$e_G \vartriangleright_\lambda c=e_G \vartriangleright_\lambda (e_G \vartriangleright_\lambda c)$ for all $c\in G$. Since $e_G \vartriangleright_\lambda -$ is invertible, we have $e_G \vartriangleright_\lambda c=c$, which implies that   \eqref{unit-d-post-group2} holds.
\end{proof}

\begin{defi}
Let $(\GGG,\{\cdot_\lambda\}_{\lambda\in \Lambda},\{\vartriangleright_\lambda\}_{\lambda\in \Lambda})$ and $(\HHH,  \{\cdot'_\lambda\}_{\lambda\in \Lambda},\{\vartriangleright'_\lambda\}_{\lambda\in \Lambda})$ be two dynamical post-groups.
A {\bf homomorphism of dynamical post-groups} from $(\GGG,\{\cdot_\lambda\}_{\lambda\in \Lambda},\{\vartriangleright_\lambda\}_{\lambda\in \Lambda})$ to $(\HHH,  \{\cdot'_\lambda\}_{\lambda\in \Lambda},\{\vartriangleright'_\lambda\}_{\lambda\in \Lambda})$ is a morphism of dynamical sets $\Psi:\GGG \to \HHH$ such that for all $\lambda\in \Lambda, a,b\in G$,
\begin{eqnarray}
%\label{homo-d-post-group1} \phi_H(\lambda,\Psi(\lambda,a))&=&\phi_G(\lambda,a),\\
\label{homo-d-post-group2} \Psi_{\lambda}(a \cdot_\lambda b)&=& \Psi_{\lambda}(a)\cdot'_{\lambda}\Psi_{\lambda}(b),\quad \Psi_{\lambda}(a \vartriangleright_\lambda b)= \Psi_{\lambda}(a)\vartriangleright'_{\lambda}\Psi_{\phi_G(\lambda,a)}(b).
\end{eqnarray}
\end{defi}

Dynamical post-groups and their homomorphisms form a category {\bf DPG}.

\begin{thm}\label{sub-adj-d-group}
Let $(\GGG, \{\cdot_\lambda\}_{\lambda\in \Lambda},\{\vartriangleright_\lambda\}_{\lambda\in \Lambda})$ be a dynamical post-group. Define $\circ_\lambda:G\times G \to G$ by
\begin{eqnarray}\label{sub-adj}
a \circ_\lambda b =a \cdot_\lambda (a \vartriangleright_\lambda b), \quad \forall a,b\in G.
\end{eqnarray}
Then  $(\GGG,\{\circ_\lambda\}_{\lambda\in \Lambda})$ is a dynamical group with $e_G$ being the unit, and the inverse map $\bar{\cdot}^\lambda:G \to G$ given by
    $$ \bar{a}^\lambda:=(L^{\rhd_\lambda}_a)^{-1} (a^\lambda),\quad \forall a\in G, $$
  where $a^\lambda$ is the inverse of $a$ in the constant dynamical group $(G, \Lambda,\{\cdot_\lambda\}_{\lambda\in \Lambda})$.

Moreover, the left multiplication $L^{\rhd_\lambda}:G\to \Perm(G)$ defined by
  $$L^{\rhd_\lambda}(a)(b)=L^{\rhd_\lambda}_ab=a \vartriangleright_{\lambda} b, \quad \forall a, b\in G, $$
  is an action of the dynamical group $(\GGG,\{\circ_\lambda\}_{\lambda\in \Lambda})$ on the constant dynamical group $(G,\Lambda,\{\cdot_\lambda\}_{\lambda\in \Lambda})$.
\end{thm}
The dynamical group $(\GGG,\{\circ_\lambda\}_{\lambda\in \Lambda})$ is called the {\bf sub-adjacent dynamical group} of the dynamical post-group
  $(\GGG,\{\cdot_\lambda\}_{\lambda\in \Lambda},\{\vartriangleright_\lambda\}_{\lambda\in \Lambda})$, and denoted by $\GGG_{\rhd}$.

\begin{proof}
 For all $\lambda\in \Lambda,a,b,c\in G$, we have
  \begin{eqnarray*}
  (a \circ_\lambda b)\circ_\lambda c
  &=& (a \circ_\lambda b)\cdot_{\lambda} ((a \circ_\lambda b)\vartriangleright_\lambda c)\\
  &=& a \cdot_{\lambda} (a \vartriangleright_\lambda b) \cdot_{\lambda} ((a \circ_\lambda b)\vartriangleright_\lambda c)\\
  &\overset{\eqref{d-post-group2}}{=}& a \cdot_{\lambda} (a \vartriangleright_\lambda b) \cdot_{\lambda} (a \vartriangleright_\lambda (b \vartriangleright_{\phi_G(\lambda,a)} c))\\
  &\overset{\eqref{d-post-group1}}{=}& a \cdot_{\lambda} \Big(a \vartriangleright_\lambda ( b \cdot_{\phi_G(\lambda,a)} (b \vartriangleright_{\phi_G(\lambda,a)} c))  \Big)\\
  &=& a \cdot_{\lambda} (a \vartriangleright_\lambda (b \circ_{\phi_G(\lambda,a)} c) )\\
  &=& a \circ_{\lambda} (b \circ_{\phi_G(\lambda,a)} c ),
  \end{eqnarray*}
  which implies that Condition {\rm(i)} in Definition \ref{dynamical-group} is satisfied.

  For all $\lambda\in \Lambda,a\in G$, we have
$$  a\circ_{\lambda}  e_G=a \cdot_{\lambda} (a \vartriangleright_\lambda e_G)\overset{\eqref{unit-d-post-group1}}{=}a \cdot_{\lambda}e_G=a;\quad
  e_G \circ_{\lambda} a=  e_G \cdot_{\lambda} (e_G \vartriangleright_{\lambda} a) \overset{\eqref{unit-d-post-group2}}{=} e_G \cdot_{\lambda} a=a,
$$
  which implies that Condition {\rm(ii)} in Definition \ref{dynamical-group} is satisfied.

  Since
  $$ a\circ_{\lambda} \bar{a}^\lambda=a \cdot_{\lambda} (a \vartriangleright_\lambda \bar{a}^\lambda)=a \cdot_{\lambda} \Big((L^{\rhd_\lambda}_a) (L^{\rhd_\lambda}_a)^{-1} (a^\lambda)\Big)=a \cdot_{\lambda} a^\lambda=e_G,$$
  and
  \begin{eqnarray*}
  a \vartriangleright_\lambda (\bar{a}^\lambda \circ_{\phi_G(\lambda,a)} a)
  &=& a \vartriangleright_\lambda ( \bar{a}^\lambda \cdot_{\phi_G(\lambda,a)}(\bar{a}^\lambda \vartriangleright_{\phi_G(\lambda,a)} a))\\
  &=& (a \vartriangleright_\lambda \bar{a}^\lambda)\cdot_{\lambda} (a \vartriangleright_\lambda (\bar{a}^\lambda \circ_{\phi_G(\lambda,a)} a))\\
  &=& a^\lambda \cdot_{\lambda} (e_G \vartriangleright_{\phi_G(\lambda,a)} a )\\
  &=& a^\lambda \cdot_{\lambda} a = e_G,
  \end{eqnarray*}
we obtain $\bar{a}^\lambda \circ_{\phi_G(\lambda,a)} a=e_G$ applying the fact that $a\vartriangleright_\lambda -$ is invertible for all $\lambda\in \Lambda,a\in G$.
Thus, Condition {\rm(iii)} in Definition \ref{dynamical-group} is satisfied. Therefore, $(\GGG,\{\circ_\lambda\}_{\lambda\in \Lambda})$ is a dynamical group.

  Eqs.\, \eqref{d-post-group1} and \eqref{d-post-group2} indicate that $L^{\rhd_\lambda}$ is an action of the dynamical group $(\GGG,\{\circ_\lambda\}_{\lambda\in \Lambda})$ on the constant dynamical group $(G,\Lambda,\{\cdot_\lambda\}_{\lambda\in \Lambda})$.
\end{proof}

\begin{pro}\label{homo-sub-d-group}
 Let $\Psi:\Lambda \times  G  \to  H$ be a homomorphism of dynamical post-groups from $  (\GGG,  \{\cdot_\lambda\}_{\lambda\in \Lambda},\{\vartriangleright_\lambda\}_{\lambda\in \Lambda}) $ to $ (\HHH, \{\cdot'_\lambda\}_{\lambda\in \Lambda},\{\vartriangleright'_\lambda\}_{\lambda\in \Lambda})$. Then $\Psi$ is a homomorphism of the sub-adjacent dynamical groups from $(\GGG,\{\circ_\lambda\}_{\lambda\in \Lambda})$ to $(\HHH, \{\circ'_\lambda\}_{\lambda\in \Lambda})$.
\end{pro}

\begin{proof}
For all $a, b\in G$, we have
  \begin{eqnarray*}
    \Psi_{\lambda}(a \circ_\lambda b)&=&  \Psi_{\lambda}(a \cdot_\lambda (a\vartriangleright_\lambda b))\\
    &\overset{\eqref{homo-d-post-group2}}{=}& \Psi_{\lambda}(a) \cdot'_\lambda \Psi_{\lambda}(a\vartriangleright_\lambda b)\\
    &\overset{\eqref{homo-d-post-group2}}{=}& \Psi_{\lambda}(a) \cdot'_\lambda
    \Psi_{\lambda}(a) \vartriangleright'_\lambda \Psi_{\phi_G(\lambda,a)}(b)\\
    &=& \Psi_{\lambda}(a) \circ'_\lambda \Psi_{\phi_G(\lambda,a)}(b),
  \end{eqnarray*}
  which implies that $\Psi$ is a homomorphism of the sub-adjacent dynamical groups.
\end{proof}

\begin{rmk}
Let $(\GGG,\{\cdot_\lambda\}_{\lambda\in \Lambda},\{\vartriangleright_\lambda\}_{\lambda\in \Lambda})$ be a weak dynamical post-group. Then we can also define $\circ_\lambda:G\times G \to G$ by \eqref{sub-adj}. In this case, $(\GGG,\{\circ_\lambda\}_{\lambda\in \Lambda})$ is just a dynamical semi-group and the left multiplication $L^{\rhd_\lambda}:G\to \Map(G)$ is an action of the dynamical semi-group $(\GGG,\{\circ_\lambda\}_{\lambda\in \Lambda})$ on the constant dynamical group $(G,\Lambda,\{\cdot_\lambda\}_{\lambda\in \Lambda})$.
\end{rmk}

At the end of this subsection,  we display the close relationship between relative Rota-Baxter operators on dynamical groups and dynamical post-groups.

\begin{pro}\label{rRBO-d-post-group}
Let $\huaB:H \to G$ be a relative Rota-Baxter operator on a  dynamical group $(\GGG,\{\circ_\lambda\}_{\lambda\in \Lambda})$ with respect to the action $\{\Phi_\lambda:G \to \Perm(H)\}_{\lambda\in\Lambda}$. Define the structure map $\phi_{\huaB}:\Lambda\times H\to H$ and the multiplication $\vartriangleright^{\huaB}_\lambda:H \times H \to H$ by
\begin{eqnarray}
\phi_{\huaB}(\lambda,x):=\phi_G(\lambda,\huaB (x)),\quad x \vartriangleright^{\huaB}_\lambda y=\Phi_{\lambda} (\huaB (x))(y),\quad \forall \lambda\in \Lambda, x,y\in H.
\end{eqnarray}
Then $(H,\Lambda,\phi_{\huaB},\{\cdot_\lambda\}_{\lambda\in \Lambda},\{\vartriangleright^{\huaB}_\lambda\}_{\lambda\in \Lambda})$ is a dynamical post-group, whose sub-adjacent dynamical group is the descendant dynamical group of $\huaB$ given in Proposition \ref{pro-descendant}.
\end{pro}

\begin{proof}
Let $\huaB:H \to G$ be a relative Rota-Baxter operator on a  dynamical group $(\GGG,\{\circ_\lambda\}_{\lambda\in \Lambda})$ with respect to the action $\{\Phi_\lambda:G \to \Perm(H)\}_{\lambda\in\Lambda}$. By Proposition \ref{pro-descendant},   $(H,\Lambda,\phi_{\huaB},\{\circ^{\huaB}_\lambda\}_{\lambda\in \Lambda})$ is a dynamical group and $\widetilde{\huaB}$ is a homomorphism of dynamical groups from $(H,\Lambda,\phi_{\huaB},\{\circ^{\huaB}_\lambda\}_{\lambda\in \Lambda})$ to $(\GGG,\{\circ_\lambda\}_{\lambda\in \Lambda})$. Since $\{\Phi_\lambda:G \to \Perm(H)\}_{\lambda\in\Lambda}$ is an action of a dynamical group $(\GGG,\{\circ_\lambda\}_{\lambda\in \Lambda})$ on a constant dynamical group $(H,\Lambda,\{\cdot_\lambda\}_{\lambda\in \Lambda})$, we deduce that $\{ (\Phi\circ \widetilde{\huaB})_{\lambda}:H  \to \Perm(H )\}_{\lambda\in \Lambda}$ is an action of the dynamical group $(H,\Lambda,\phi_{\huaB},\{\circ^{\huaB}_\lambda\}_{\lambda\in \Lambda})$ on the constant dynamical group $(H,\Lambda,\{\cdot_\lambda\}_{\lambda\in \Lambda})$. Moreover, for all $x\in H$, we have
\vspace{-.2cm}
$$L^{\rhd_\lambda}(x)=(\Phi \circ \widetilde{\huaB})_{\lambda}(x)=\Phi_{\lambda}\circ \widetilde{\huaB}_{\lambda} (x)=\Phi_{\lambda} (\huaB (x)),$$
which implies that is $L^{\rhd_\lambda}$ an action of $(H,\Lambda,\phi_{\huaB},\{\circ^{\huaB}_\lambda\}_{\lambda\in \Lambda})$ on $(H,\Lambda,\{\cdot_\lambda\}_{\lambda\in \Lambda})$. Therefore, we deduce that $(H,\Lambda,\phi_{\huaB},\{\cdot_\lambda\}_{\lambda\in \Lambda},\{\vartriangleright_\lambda\}_{\lambda\in \Lambda})$ is a dynamical post-group by Theorem \ref{sub-adj-d-group}. The other conclusion can be easily obtained.
\end{proof}

\begin{pro}\label{Id-rRBO}
	Let $(\GGG,\{\cdot_\lambda\}_{\lambda\in \Lambda},\{\vartriangleright_\lambda\}_{\lambda\in \Lambda})$ be a dynamical post-group. Then the identity map $\Id:G \to G$ is a
	relative Rota-Baxter operator on the sub-adjacent dynamical group $(\GGG,\{\circ_\lambda\}_{\lambda\in \Lambda})$ given in Theorem \ref{sub-adj-d-group} with respect to the action $L^{\rhd_\lambda}:  G \to \Perm(G)$.
\end{pro}

\begin{proof}
	By \eqref{sub-adj}, we get
\vspace{-.2cm}
	$$ \Id(a)\circ_{\lambda} \Id(b)=\Id(a \cdot_{\lambda} L^{\rhd_\lambda}_a (b)),\quad \forall \lambda\in\Lambda,a,b\in G, $$
	which implies that $\Id:G \to G$ is a relative Rota-Baxter operator on $(\GGG,\{\circ_\lambda\}_{\lambda\in \Lambda})$ with respect to the action $L^{\rhd_\lambda}$.
\end{proof}
\vspace{-.3cm}

\subsection{Dynamical post-groups and braided dynamical groups}

Now we show that the category of dynamical post-groups is isomorphic to the category of braided dynamical groups. Consequently, we obtain that a dynamical post-group gives rise to a solution of the dynamical Yang-Baxter equation.
\vspace{-.2cm}

\begin{thm}\label{DPG-BDG}
Let $(\GGG,\{\cdot_\lambda\}_{\lambda\in \Lambda},\{\vartriangleright_\lambda\}_{\lambda\in \Lambda})$ be a dynamical post-group. Then $(\GGG_{\rhd},R_G)$ is a braided dynamical group, where $\GGG_\rhd=(G,\Lambda,\phi_G,\{\circ_\lambda\}_{\lambda\in \Lambda})$ is the sub-adjacent dynamical group, and $R_G:\Lambda \times G \times G \to G \times G$ is defined as follows:
\vspace{-.2cm}
\begin{eqnarray}\label{RG}
  R_G(\lambda,a,b)=(a \vartriangleright_\lambda b,\overline{(a \vartriangleright_\lambda b)}^{\lambda} \circ_{\phi_G(\lambda,a \vartriangleright_\lambda b)} (a \circ_{\lambda} b) ),\quad \forall \lambda \in \Lambda,~a,b\in G.
\end{eqnarray}

\vspace{-.2cm}
Moreover, let $\Psi$ be a homomorphism of dynamical post-groups from $(\GGG,\{\cdot_\lambda\}_{\lambda\in \Lambda},\{\vartriangleright_\lambda\}_{\lambda\in \Lambda})$ to $(\HHH,  \{\cdot'_\lambda\}_{\lambda\in \Lambda},\{\vartriangleright'_\lambda\}_{\lambda\in \Lambda})$. Then $\Psi$ is a homomorphism of braided dynamical groups from $(\GGG_\rhd,R_G)$ to $(\HHH_{\rhd'},R_H)$.
\end{thm}

\begin{proof}
By Proposition \ref{Id-rRBO}, the identity map $\Id:G \to G$ is a
relative Rota-Baxter operator on the sub-adjacent dynamical group $(\GGG,\{\circ_\lambda\}_{\lambda\in \Lambda})$ with respect to the action $L^{\rhd_\lambda}:  (G \to \Perm(G)$. Moreover, by Proposition \ref{rRBO-mp}, $(\GGG_\rhd,\GGG_\rhd,R_G)$ is a matched pair, where
\vspace{-.3cm}
\begin{eqnarray*}
a\overset{\lambda}{\rightharpoonup} b&=&a \vartriangleright_\lambda b;\\
a\overset{\phi_G(\lambda,a)}{\leftharpoonup\joinrel\relbar\joinrel\relbar\joinrel\relbar} b&=&\overline{(a \vartriangleright_\lambda b)}^{\lambda} \circ_{\phi_G(\lambda,a \vartriangleright_\lambda b)} (a \circ_{\lambda} b), \quad \forall \lambda\in \Lambda,a,b\in G.
\end{eqnarray*}
Furthermore, we have
\vspace{-.3cm}
$$  (a\overset{\lambda}{\rightharpoonup} b)\circ_{\lambda} (a\overset{\phi_G(\lambda,a)}{\leftharpoonup\joinrel\relbar\joinrel\relbar\joinrel\relbar} b)=
(a \vartriangleright_\lambda b)\circ_{\lambda}\Big(\overline{(a \vartriangleright_\lambda b)}^{\lambda} \circ_{\phi_G(\lambda,a \vartriangleright_\lambda b)} (a \circ_{\lambda} b) \Big)=a \circ_{\lambda} b,$$
which implies that   \eqref{braided-com} holds. Therefore,  $(\GGG_\rhd,R_G)$ is a braided dynamical group and $R_G$ is a solution of the dynamical Yang-Baxter equation on the dynamical set $\GGG$ by Theorem \ref{braided-solution}.

Let $\Psi$ be a homomorphism  from $(\GGG,\{\cdot_\lambda\}_{\lambda\in \Lambda},\{\vartriangleright_\lambda\}_{\lambda\in \Lambda})$ to $(\HHH,  \{\cdot'_\lambda\}_{\lambda\in \Lambda},\{\vartriangleright'_\lambda\}_{\lambda\in \Lambda})$. By Proposition \ref{homo-sub-d-group}, $\Psi$ is a homomorphism of the corresponding sub-adjacent dynamical groups from $(\GGG,\{\circ_\lambda\}_{\lambda\in \Lambda})$ to $(\HHH, \{\circ'_\lambda\}_{\lambda\in \Lambda})$. Moreover, for all $\lambda\in \Lambda,a,b\in G$, we have
{\small
\begin{eqnarray*}
&&(\Psi\times \Psi)_{\lambda} \circ \sigma_{\lambda} (a,b)\\
&=& (\Psi\times \Psi)_{\lambda} \Big(a \vartriangleright_\lambda b,\overline{ a \vartriangleright_\lambda b}^{\lambda} \circ_{\phi_G(\lambda,a \vartriangleright_\lambda b)}      (a \circ_{\lambda} b) \Big)\\
&=& \Big( \Psi_{\lambda}(a \vartriangleright_\lambda b),\Psi_{\phi_G(\lambda,a \vartriangleright_\lambda b )} \Big( \overline{a \vartriangleright_\lambda b}^{\lambda} \circ_{\phi_G(\lambda,a \vartriangleright_\lambda b)} (a\circ_{\lambda} b) \Big)  \Big)\\
&\overset{\eqref{homo-d-group}\eqref{homo-d-post-group2}}{=}& \Big( \Psi_{\lambda}(a) \vartriangleright'_\lambda  \Psi_{\phi_G(\lambda,a)}(b),
\Psi_{\phi_G(\lambda,a \vartriangleright_\lambda b )}(\overline{a \vartriangleright_\lambda b}^{\lambda}) \circ'_{\phi_{H}(\lambda,\Psi_{\lambda}(a \vartriangleright_\lambda b))} (\Psi_{\lambda}(a \circ_{\lambda} b)) \Big)\\
&=& \Big( \Psi_{\lambda}(a) \vartriangleright'_\lambda  \Psi_{\phi_G(\lambda,a)}(b),\overline{\Psi_{\lambda}(a) \vartriangleright'_\lambda  \Psi_{\phi_G(\lambda,a)}(b)}^{\lambda} \circ'_{\phi_{H}(\lambda,\Psi_{\lambda}(a)\vartriangleright'_\lambda \Psi_{\phi_G(\lambda,a)}(b) )} (\Psi_{\lambda}(a) \circ'_{\lambda} \Psi_{\phi_G(\lambda,a)}(b) )\\
&=& \sigma'_{\lambda} (\Psi_{\lambda}(a),\Psi_{\phi_G(\lambda,a)}(b) )\\
&=& \sigma'_{\lambda} \circ (\Psi\times \Psi)_{\lambda} (a,b),
\end{eqnarray*}
}
which implies that $\Psi$ is a homomorphism  from $(\GGG_{\rhd},R_G)$ to $(\HHH_{\rhd'},R_H)$.
\end{proof}

\begin{cor}
Let $(\GGG,\{\cdot_\lambda\}_{\lambda\in \Lambda},\{\vartriangleright_\lambda\}_{\lambda\in \Lambda})$ be a dynamical post-group. Then
 $R_G$ is a solution of the dynamical Yang-Baxter equation on the dynamical set $\GGG=(G,\Lambda,\phi_G)$.
\end{cor}
\begin{proof}
  It follows from Theorem \ref{braided-solution} and Theorem \ref{DPG-BDG} directly.
\end{proof}

Since a relative Rota-Baxter operator on a dynamical group naturally gives rise to a dynamical post-group, a relative Rota-Baxter operator on a dynamical group naturally provides a solution of the dynamical Yang-Baxter equation.

\begin{cor}
Let $\huaB:H \to G$ be a relative Rota-Baxter operator on a  dynamical group $(\GGG,\{\circ_\lambda\}_{\lambda\in \Lambda})$ with respect to an action $\{\Phi_\lambda:G \to \Perm(H)\}_{\lambda\in\Lambda}$. Then $R_{\huaB}:\Lambda \times H \times H
\to H \times H $ defined by
\vspace{-.2cm}
\begin{eqnarray}\label{rRBO-solution}
\quad R_{\huaB}(\lambda,x,y)=\Big(\Phi_{\lambda} (\huaB (x))(y),\overline{\Phi_{\lambda} (\huaB (x))(y)}^{\lambda} \circ^{\huaB}_{\phi_{\huaB}(\lambda,\Phi_{\lambda} (\huaB (x))(y))} (x \circ^{\huaB}_{\lambda} y) \Big), \quad \forall \lambda\in \Lambda,x,y\in H,
\end{eqnarray}
is a solution of the dynamical Yang-Baxter equation on the dynamical group $(H,\Lambda,\phi_{\huaB},\{\circ^{\huaB}_\lambda\}_{\lambda\in \Lambda})$ given in Proposition \ref{pro-descendant}.
\end{cor}

\vspace{-.3cm}
Let $(\GGG,\{\circ_\lambda\}_{\lambda\in \Lambda},R_G)$ be a braided dynamical group. Write
$R_G(a,b)=(a\overset{\lambda}{\rightharpoonup} b,a\overset{\phi_G(\lambda,a)}{\leftharpoonup\joinrel\relbar\joinrel\relbar\joinrel\relbar} b)$, and define $\vartriangleright_{\lambda}:G \times G \to G$ and $\cdot_{\lambda}:G \times G \to G$   by
\vspace{-.2cm}
$$  a \vartriangleright_{\lambda} b=a\overset{\lambda}{\rightharpoonup} b, \quad
 a \cdot_{\lambda} b=a \circ_{\lambda} ( \bar{a}^{\lambda} \overset{\phi_G(\lambda,a)}{\joinrel\relbar\joinrel\relbar\joinrel\rightharpoonup} b), $$
for all  $\lambda\in \Lambda$, $a,b\in G$. Here $\bar{a}^{\lambda}$ is the inverse of $a$ in the dynamical group $(G,\Lambda,\phi_{G},\{\circ_\lambda\}_{\lambda\in \Lambda})$.

\begin{pro}\label{BDG-DPG}
With above notations,   $(\GGG,\{\cdot_\lambda\}_{\lambda\in \Lambda},\{\vartriangleright_\lambda\}_{\lambda\in \Lambda})$ is a dynamical post-group.

Moreover, let $\Psi$ be a homomorphism of braided  dynamical groups from $(\GGG,\{\circ_\lambda\}_{\lambda\in \Lambda},R_G)$ to $(\HHH,\{\circ'_\lambda\}_{\lambda\in \Lambda},R_H)$. Then $\Psi$ is a homomorphism of the corresponding dynamical post-groups from $(\GGG,\{\cdot_\lambda\}_{\lambda\in \Lambda},\{\vartriangleright_\lambda\}_{\lambda\in \Lambda})$ to $(\HHH,  \{\cdot'_\lambda\}_{\lambda\in \Lambda},\{\vartriangleright'_\lambda\}_{\lambda\in \Lambda})$.
\end{pro}

\begin{proof}
For all $\lambda \in \Lambda,a,b,c\in G$, we  have
\vspace{-.2cm}
{\small
\begin{eqnarray*}
a\cdot_{\lambda} (b \cdot_{\lambda} c)
&=& a \circ_{\lambda} (\bar{a}^{\lambda}  \overset{\phi_G(\lambda,a)}{\relbar\joinrel\relbar\joinrel\relbar\joinrel\rightharpoonup} (b \cdot_{\lambda} c) )\\
&=& a \circ_{\lambda} \Big( (\bar{a}^{\lambda}  \overset{\phi_G(\lambda,a)}{\relbar\joinrel\relbar\joinrel\relbar\joinrel\rightharpoonup} b)
\cdot_{\phi_G(\lambda,a)} (\bar{a}^{\lambda}  \overset{\phi_G(\lambda,a)}{\relbar\joinrel\relbar\joinrel\relbar\joinrel\rightharpoonup} c)\Big)\\
&=& a \circ_{\lambda} \Big( (\bar{a}^{\lambda}  \overset{\phi_G(\lambda,a)}{\relbar\joinrel\relbar\joinrel\relbar\joinrel\rightharpoonup} b)
\circ_{\phi_G(\lambda,a)}  \Big(\overline{(\bar{a}^{\lambda}  \overset{\phi_G(\lambda,a)}{\relbar\joinrel\relbar\joinrel\relbar\joinrel\rightharpoonup} b)}^{\phi_G(\lambda,a)}\overset{\phi_G(\lambda,a \cdot_{\lambda} b )}
{\relbar\joinrel\relbar\joinrel\relbar\joinrel\rightharpoonup}(\bar{a}^{\lambda}  \overset{\phi_G(\lambda,a)}{\relbar\joinrel\relbar\joinrel\relbar\joinrel\rightharpoonup} c)\Big)\Big)\\
&=& \Big(a \circ_{\lambda} ( \bar{a}^{\lambda} \overset{\phi_G(\lambda,a)}{\joinrel\relbar\joinrel\relbar\joinrel\rightharpoonup} b)\Big)\circ_{\lambda} \Big( \Big(\overline{(\bar{a}^{\lambda}  \overset{\phi_G(\lambda,a)}{\relbar\joinrel\relbar\joinrel\relbar\joinrel\rightharpoonup} b)}^{\phi_G(\lambda,a)} \circ_{\phi_G(\lambda,a \cdot_{\lambda} b )}  \bar{a}^{\lambda} \Big)  \overset{\phi_G(\lambda,a \cdot_{\lambda} b )}{\relbar\joinrel\relbar\joinrel\relbar\joinrel\rightharpoonup} c  \Big)\\
&=& (a \cdot_{\lambda} b) \circ_{\lambda} \Big(  \overline{(a \cdot_{\lambda} b)}^{\lambda} \overset{\phi_G(\lambda,a \cdot_{\lambda} b )}{\relbar\joinrel\relbar\joinrel\relbar\joinrel\rightharpoonup} c  \Big)\\                             &=& (a\cdot_{\lambda} b) \cdot_{\lambda} c.
\end{eqnarray*}
}
\vspace{-.2cm}
By \eqref{mp-d-groups-4} and \eqref{mp-d-groups-5}, we can deduce that $a\overset{\lambda}{\rightharpoonup} e_G= e_G$ for all $\lambda\in \Lambda, a\in G$. Then we have
$$ a\cdot_{\lambda} e_G=a \circ_{\lambda}( \bar{a}^{\lambda} \overset{\phi_G(\lambda,a)}{\joinrel\relbar\joinrel\relbar\joinrel\rightharpoonup} e_G)=a, $$
\vspace{-.3cm}
$$ a \cdot_{\lambda}( a \overset{\lambda}{\rightharpoonup} \bar{a}^{\lambda})=
a \circ_{\lambda}\Big( \bar{a}^{\lambda} \overset{\phi_G(\lambda,a)}{\joinrel\relbar\joinrel\relbar\joinrel\rightharpoonup}(a \overset{\lambda}{\rightharpoonup} \bar{a}^{\lambda})\Big)\overset{\eqref{mp-d-groups-2}}{=}
a \circ_{\lambda} \Big((\bar{a}^{\lambda} \circ_{\phi_G(\lambda,a)} a)\overset{\phi_G(\lambda,a)}{\joinrel\relbar\joinrel\relbar\joinrel\rightharpoonup} \bar{a}^{\lambda} \Big)\overset{\eqref{mp-d-groups-1}}{=} a \circ_{\lambda} e_G=a. $$
Thus, $e_G$ is a right unit of the multiplication $\cdot_{\lambda}$ and $a \overset{\lambda}{\rightharpoonup} \bar{a}^{\lambda}$ is a right inverse of $a$ with respect to $\cdot_{\lambda}$ for all $\lambda \in \Lambda$.
Therefore,   $(G, \Lambda,\{\cdot_\lambda\}_{\lambda\in \Lambda})$ is a constant dynamical group.

Furthermore, for all $a,b,c\in G$, we have
\vspace{-.2cm}
{\small
\begin{eqnarray*}
a \vartriangleright_\lambda (b \cdot_{\phi_G(\lambda,a)} c)&=& a\overset{\lambda}{\rightharpoonup} (b \cdot_{\phi_G(\lambda,a)} c)\\
&=& a\overset{\lambda}{\rightharpoonup} \Big(b \circ_{\phi_G(\lambda,a)} (\bar{b}^{\phi_G(\lambda,a)}  \overset{\phi_G(\phi_G(\lambda,a),b)}{\joinrel\relbar\joinrel\relbar\joinrel\relbar\joinrel\relbar
\joinrel\rightharpoonup}  c) \Big)\\
&\overset{\eqref{mp-d-groups-3}}{=}& (a\overset{\lambda}{\rightharpoonup} b) \circ_{\lambda}
\Big( (a\overset{\phi_G(\lambda,a)}{\leftharpoonup\joinrel\relbar\joinrel\relbar\joinrel\relbar} b) \overset{\phi_G(\lambda,a\overset{\lambda}{\rightharpoonup} b)} {\relbar\joinrel\relbar\joinrel\relbar\joinrel\relbar\joinrel\rightharpoonup}  (\bar{b}^{\phi_G(\lambda,a)}  \overset{\phi_G(\lambda,a \circ_{\lambda}
b)}{\joinrel\relbar\joinrel\relbar\joinrel\relbar\joinrel\relbar
\joinrel\rightharpoonup}  c)  \Big)\\
&\overset{\eqref{braided-com}}{=}& (a\overset{\lambda}{\rightharpoonup} b) \circ_{\lambda}
\Big( (a\overset{\phi_G(\lambda,a)}{\leftharpoonup\joinrel\relbar\joinrel\relbar\joinrel\relbar} b) \overset{\phi_G(\lambda,a\overset{\lambda}{\rightharpoonup} b)} {\relbar\joinrel\relbar\joinrel\relbar\joinrel\relbar\joinrel\rightharpoonup}  (\bar{b}^{\phi_G(\lambda,a)}  \overset{\phi_G (\phi_G(\lambda,a\overset{\lambda}{\rightharpoonup} b),a\overset{\phi_G(\lambda,a)}{\leftharpoonup\joinrel\relbar\joinrel\relbar\joinrel\relbar} b)}{\joinrel\relbar\joinrel\relbar\joinrel\relbar\joinrel\relbar\joinrel\relbar\joinrel\relbar
\joinrel\relbar\joinrel\relbar\joinrel\relbar\joinrel\relbar\joinrel\relbar\joinrel\relbar\joinrel
\rightharpoonup}  c)  \Big)\\
&\overset{\eqref{mp-d-groups-2}}{=}& (a\overset{\lambda}{\rightharpoonup} b) \circ_{\lambda}
\Big( ((a\overset{\phi_G(\lambda,a)}{\leftharpoonup\joinrel\relbar\joinrel\relbar\joinrel\relbar} b) \circ_{\phi_G(\lambda,a\overset{\lambda}{\rightharpoonup} b)} \bar{b}^{\phi_G(\lambda,a)}) \overset{\phi_G(\lambda,a\overset{\lambda}{\rightharpoonup} b)} {\relbar\joinrel\relbar\joinrel\relbar\joinrel\relbar\joinrel\rightharpoonup} c   \Big)\\
&=& (a\overset{\lambda}{\rightharpoonup} b) \circ_{\lambda} \Big( \overline{(a\overset{\lambda}{\rightharpoonup} b)}^{\lambda}  \overset{\phi_G(\lambda,a\overset{\lambda}{\rightharpoonup} b)} {\relbar\joinrel\relbar\joinrel\relbar\joinrel\relbar\joinrel\rightharpoonup} c    \Big)\\
&=& (a\overset{\lambda}{\rightharpoonup} b) \cdot_{\lambda} (a\overset{\lambda}{\rightharpoonup} c)\\
&=& (a \vartriangleright_\lambda b) \cdot_{\lambda} (a \vartriangleright_\lambda c).
\end{eqnarray*}
}
Thus, we obtain that $a \vartriangleright_\lambda (b \cdot_{\phi_G(\lambda,a)} c)=(a \vartriangleright_\lambda b) \cdot_{\lambda} (a \vartriangleright_\lambda c).$

Moreover, for all  $\lambda \in \Lambda$ and $a,b,c\in G$, we have
\vspace{-.3cm}
\begin{eqnarray*}
a \vartriangleright_\lambda (b \vartriangleright_{\phi_G(\lambda,a)} c)&=&
a \overset{\lambda}{\rightharpoonup} (b\overset{\phi_G(\lambda,a)}{\relbar\joinrel\relbar\joinrel\relbar\joinrel\rightharpoonup}   c)\\
&\overset{\eqref{mp-d-groups-2}}{=}&  (a \circ_{\lambda} b) \overset{\lambda}{\rightharpoonup} c\\
&=& (a \circ_{\lambda} b) \vartriangleright_\lambda c\\
&=& (a \cdot_{\lambda} (a \vartriangleright_\lambda b)) \vartriangleright_\lambda c,
\end{eqnarray*}
which implies that $(\GGG,\{\cdot_\lambda\}_{\lambda\in \Lambda},\{\vartriangleright_\lambda\}_{\lambda\in \Lambda})$ is a dynamical post-group.

Let $\Psi$ be a homomorphism of braided  dynamical groups from $(\GGG,\{\circ_\lambda\}_{\lambda\in \Lambda},R_G)$ to $(\HHH,\{\circ'_\lambda\}_{\lambda\in \Lambda},R_H)$. Then we have
\vspace{-.3cm}
$$ (\Psi \times \Psi)_{\lambda} \circ \sigma_{\lambda}=\sigma'_{\lambda} \circ (\Psi \times \Psi)_{\lambda}.$$
That means that for all $\lambda\in \Lambda$ and $a,b\in G$, we have
\vspace{-.3cm}
\begin{eqnarray}
\label{homo-1} \Psi_{\lambda} (a \overset{\lambda}{\rightharpoonup} b)&=&\Psi_{\lambda} (a) \overset{\lambda}{\rightharpoonup} \Psi_{\phi_G(\lambda,a)} (b),\\
\label{homo-2} \Psi_{\phi_G(\lambda,a\overset{\lambda}{\rightharpoonup} b)} (a
\overset{\phi_G(\lambda,a)}{\leftharpoonup\joinrel\relbar\joinrel\relbar\joinrel\relbar} b)&=&
\Psi_{\lambda} (a) \overset{\phi_G(\lambda,\Psi_{\lambda} (a))}{\leftharpoonup\joinrel\relbar\joinrel\relbar\joinrel\relbar}\Psi_{\phi_G(\lambda,a)} (b).
\end{eqnarray}
Therefore, for all  $\lambda \in \Lambda$ and $a,b,c\in G$, we have
\vspace{-.3cm}
\begin{eqnarray*}
\Psi_{\lambda}(a \vartriangleright_\lambda b)= \Psi_{\lambda}(a \overset{\lambda}{\rightharpoonup} b)
\overset{\eqref{homo-1}}{=} \Psi_{\lambda}(a) \overset{\lambda}{\rightharpoonup} \Psi_{\phi_G(\lambda,a)}(b)
= \Psi_{\lambda}(a) \vartriangleright'_\lambda \Psi_{\phi_G(\lambda,a)}(b),
\end{eqnarray*}
and
\vspace{-.3cm}
{\small
\begin{eqnarray*}
\Psi_{\lambda}(a \cdot_\lambda b)&=& \Psi_{\lambda}(a \circ_{\lambda} (\bar{a}^{\lambda}  \overset{\phi_G(\lambda,a)}{\relbar\joinrel\relbar\joinrel\relbar\joinrel\rightharpoonup} b) )\\
&=& \Psi_{\lambda}(a) \circ'_{\lambda}  \Psi_{\phi_G(\lambda,a)}(\bar{a}^{\lambda}  \overset{\phi_G(\lambda,a)}{\relbar\joinrel\relbar\joinrel\relbar\joinrel\rightharpoonup} b)\\
&\overset{\eqref{homo-1}}{=}& \Psi_{\lambda}(a) \circ'_{\lambda} \Big( \Psi_{\phi_G(\lambda,a)}(\bar{a}^{\lambda})\overset{\phi_G(\lambda,a)}
{\relbar\joinrel\relbar\joinrel\relbar\joinrel\rightharpoonup}   \Psi_{\phi_G(\phi_G(\lambda,a),\bar{a}^{\lambda})}(b)   \Big)\\
&\overset{\eqref{homo-d-group}}{=}& \Psi_{\lambda}(a) \circ'_{\lambda} \Big(    \overline{\Psi_{\lambda}(a)}^{\lambda} \overset{\phi_H(\lambda,\Psi_{\lambda}(a))}{\relbar\joinrel\relbar\joinrel\relbar\joinrel\relbar
\joinrel\rightharpoonup} \Psi_{\lambda}(b) \Big)\\
&=& \Psi_{\lambda}(a) \cdot'_{\lambda} \Psi_{\lambda}(b),
\end{eqnarray*}
}
which implies that $\Psi$ is a homomorphism of the corresponding dynamical post-groups from $(\GGG,\{\cdot_\lambda\}_{\lambda\in \Lambda},\{\vartriangleright_\lambda\}_{\lambda\in \Lambda})$ to $(\HHH,  \{\cdot'_\lambda\}_{\lambda\in \Lambda},\{\vartriangleright'_\lambda\}_{\lambda\in \Lambda})$.
\end{proof}
\vspace{-.3cm}

\begin{thm}
Theorem \ref{DPG-BDG} gives a functor $\huaY:{\bf DPG} \to {\bf BDG}$ and Proposition \ref{BDG-DPG} defines a functor $\huaP:{\bf BDG} \to {\bf DPG}$ that is the inverse of $\huaY$, giving an isomorphism between the categories ${\bf DPG}$ and ${\bf BDG}$.
\end{thm}

\begin{proof}
Let $(\GGG,\{\cdot_\lambda\}_{\lambda\in \Lambda},\{\vartriangleright_\lambda\}_{\lambda\in \Lambda})$ and $(\HHH,  \{\cdot'_\lambda\}_{\lambda\in \Lambda},\{\vartriangleright'_\lambda\}_{\lambda\in \Lambda})$ be two dynamical post-groups and $\Psi$ be a homomorphism of dynamical post-groups from
$(\GGG,\{\cdot_\lambda\}_{\lambda\in \Lambda},\{\vartriangleright_\lambda\}_{\lambda\in \Lambda})$ to $(\HHH,  \{\cdot'_\lambda\}_{\lambda\in \Lambda},\{\vartriangleright'_\lambda\}_{\lambda\in \Lambda})$. By Theorem \ref{DPG-BDG} and Proposition \ref{BDG-DPG}, we have
$$ (\huaP\huaY)(\GGG,\{\cdot_\lambda\}_{\lambda\in \Lambda},\{\vartriangleright_\lambda\}_{\lambda\in \Lambda})=(\GGG,\{\cdot_\lambda\}_{\lambda\in \Lambda},\{\vartriangleright_\lambda\}_{\lambda\in \Lambda}),\quad (\huaP\huaY)(\Psi)=\Psi.$$
Thus, $\huaP\huaY=\Id_{\bf DPG}$. Similarly, we also have $\huaY\huaP=\Id_{\bf BDG}$. Therefore,
the functor $\huaY$ is the inverse of the functor $\huaP$.
\end{proof}
\vspace{-.3cm}

\section{Dynamical skew-braces}\label{sec:d-skew-brace}
We first introduce the notion of a dynamical skew brace. Then we show that the category of dynamical skew braces is isomorphic to the category of dynamical post-groups, and therefore isomorphic to the category of braided dynamical groups.

In this section, as before, $\mathscr{G}=(G,\Lambda,\phi_G)$ is a dynamical set.

\begin{defi} A triple
$(\GGG,\{\cdot_\lambda\}_{\lambda\in \Lambda},\{\circ_\lambda\}_{\lambda\in \Lambda})$ is called a {\bf dynamical skew brace}, if
\begin{itemize}
\item[{\rm(i)}] $(G, \Lambda,\{\cdot_\lambda\}_{\lambda\in \Lambda})$ is a constant dynamical group;

\item[{\rm(ii)}] $(\GGG,\{\circ_\lambda\}_{\lambda\in \Lambda})$ is a dynamical group;

\item[{\rm(iii)}] $\cdot$ and $\circ$ satisfy the following compatibility condition:
\begin{equation}\label{d-skew-brace-equ}
  a \circ_{\lambda} (b \cdot_{\phi_{G}(\lambda,a)} c)=(a \circ_{\lambda} b)\cdot_{\lambda} a^{\lambda} \cdot_{\lambda} (a \circ_{\lambda} c), \quad \forall a,b,c\in G.
\end{equation}
\end{itemize}
A {\bf dynamical brace} is a dynamical skew brace $(\GGG,\{\cdot_\lambda\}_{\lambda\in \Lambda},\{\circ_\lambda\}_{\lambda\in \Lambda})$ in which all groups $(G,\cdot_\lambda)$ are abelian for all $\lambda\in \Lambda.$
\end{defi}

\begin{rmk}
This notion of a dynamical skew brace is different from the one given in \cite[Definition 4.9]{Ferri}. First, we generalize a group $(G,\cdot)$ used in \cite[Definition 4.9]{Ferri} to a constant dynamical group $(G, \Lambda,\{\cdot_\lambda\}_{\lambda\in \Lambda})$. Second, we use the dynamical group, in which there is a unit $e_G$ satisfying $a \circ_\lambda e_G=e_G \circ_\lambda a=a$, rather than the left quasi-group used in \cite{Ferri,Matsumoto}. Note that the condition $a \circ_\lambda e_G=e_G \circ_\lambda a=a$ is called the zero-symmetric condition in \cite[Definition 3.4]{Matsumoto}, which is    equivalent to the condition $\phi_G(\lambda,e_G)=\lambda$ (\cite[Definition 4]{Rump}).
\end{rmk}

\begin{lem}
Let $(\GGG,\{\cdot_\lambda\}_{\lambda\in \Lambda},\{\circ_\lambda\}_{\lambda\in \Lambda})$ be a dynamical skew brace. Then for all $\lambda\in \Lambda, a,b\in G$, we have
\vspace{-.3cm}
\begin{eqnarray}\label{equ-d-brace}
(a \circ_{\lambda} b)^{\lambda}= a^{\lambda} \cdot_{\lambda} (a \circ_{\lambda} b^{\phi_G(\lambda,a)}) \cdot_{\lambda} a^{\lambda}.
\end{eqnarray}
\end{lem}

\begin{proof}
For all $\lambda\in \Lambda, a,b\in G$, by \eqref{d-skew-brace-equ} and the fact that $(G, \Lambda,\{\cdot_\lambda\}_{\lambda\in \Lambda})$ is a constant dynamical group, we have
\vspace{-.3cm}
\begin{eqnarray*}
(a \circ_{\lambda} b) \cdot_{\lambda} a^{\lambda} \cdot_{\lambda} (a \circ_{\lambda} b^{\phi_G(\lambda,a)}) \cdot_{\lambda} a^{\lambda}
&=& \Big( a \circ_{\lambda} (b  \cdot_{\phi_G(\lambda,a)} b^{\phi_G(\lambda,a)})
\Big)\cdot_{\lambda} a^{\lambda}\\
&=& (a  \circ_{\lambda} e_G)\cdot_{\lambda} a^{\lambda}= e_G\\
&=& (a \circ_{\lambda} b) \cdot_{\lambda} (a \circ_{\lambda} b)^{\lambda},
\end{eqnarray*}
which implies that $(a \circ_{\lambda} b)^{\lambda}= a^{\lambda} \cdot_{\lambda} (a \circ_{\lambda} b^{\phi_G(\lambda,a)}) \cdot_{\lambda} a^{\lambda}.$
\end{proof}

\begin{defi}
Let $(\GGG,\{\cdot_\lambda\}_{\lambda\in \Lambda},\{\circ_\lambda\}_{\lambda\in \Lambda})$ and
$(\HHH,\{\cdot'_\lambda\}_{\lambda\in \Lambda},\{\circ'_\lambda\}_{\lambda\in \Lambda})$ be two
dynamical skew braces. A {\bf homomorphism of dynamical skew braces} from
$(\GGG,\{\cdot_\lambda\}_{\lambda\in \Lambda},\{\circ_\lambda\}_{\lambda\in \Lambda})$ to $(\HHH,\{\cdot'_\lambda\}_{\lambda\in \Lambda},\{\circ'_\lambda\}_{\lambda\in \Lambda})$ is a morphism of dynamical sets $\Psi:\GGG \to \HHH$ such that for all $\lambda\in \Lambda, a,b\in G,$
\vspace{-.1cm}
\begin{eqnarray}
%\label{homo-d-skew-brace1} \phi_H(\lambda,\Psi(\lambda,a))&=&\phi_G(\lambda,a),\\
\label{homo-d-skew-brace2} \Psi_{\lambda}(a \cdot_\lambda b)&=& \Psi_{\lambda}(a)\cdot'_{\lambda}\Psi_{\lambda}(b),\quad \Psi_{\lambda}(a \circ_\lambda b)= \Psi_{\lambda}(a)\circ'_{\lambda}\Psi_{\phi_G(\lambda,a)}(b).
\end{eqnarray}
\end{defi}

Dynamical skew braces and their homomorphisms form a category {\bf DSB}.

By Theorem \ref{sub-adj-d-group}, a dynamical post-group
$(\GGG,\{\cdot_\lambda\}_{\lambda\in \Lambda},\{\vartriangleright_\lambda\}_{\lambda\in \Lambda})$ gives rise to the sub-adjacent dynamical group $(\GGG,\{\circ_\lambda\}_{\lambda\in \Lambda})$. Moreover, it also gives rise to a dynamical skew brace.

\begin{pro}\label{dpg-dsb}
Let $(\GGG,\{\cdot_\lambda\}_{\lambda\in \Lambda},\{\vartriangleright_\lambda\}_{\lambda\in \Lambda})$ be a dynamical post-group. Then   $(\GGG,\{\cdot_\lambda\}_{\lambda\in \Lambda},\{\circ_\lambda\}_{\lambda\in \Lambda})$ is a dynamical skew brace, where $\{\circ_\lambda\}_{\lambda\in \Lambda}$ is given by \eqref{sub-adj}.

Moreover, let $\Psi:\GGG\to\HHH$ be a homomorphism of dynamical post-groups from $(\GGG,\{\cdot_\lambda\}_{\lambda\in \Lambda},\{\vartriangleright_\lambda\}_{\lambda\in \Lambda})$ to $ (\HHH,\{\cdot'_\lambda\}_{\lambda\in \Lambda},\{\vartriangleright'_\lambda\}_{\lambda\in \Lambda})$. Then $\Psi$ induces a homomorphism of the corresponding dynamical skew braces from $(\GGG,\{\cdot_\lambda\}_{\lambda\in \Lambda},\{\circ_\lambda\}_{\lambda\in \Lambda})$ to $(\HHH,\{\cdot'_\lambda\}_{\lambda\in \Lambda},\{\circ'_\lambda\}_{\lambda\in \Lambda})$.
\end{pro}

\begin{proof}
By the definition of the sub-adjacent dynamical group $(\GGG,\{\circ_\lambda\}_{\lambda\in \Lambda})$, we have
\vspace{-.2cm}
\begin{eqnarray*}
a \circ_{\lambda} (b \cdot_{\phi_G(\lambda,a)} c)
&=& a \cdot_{\lambda} (a \vartriangleright_\lambda (b \cdot_{\phi_G(\lambda,a)} c))\\
&\overset{\eqref{d-post-group1}}{=}& a \cdot_{\lambda} (a \vartriangleright_\lambda b) \cdot_{\lambda} (a \vartriangleright_\lambda c)\\
&=& a \cdot_{\lambda} (a \vartriangleright_\lambda b) \cdot_{\lambda}  (a^{\lambda} \cdot_{\lambda} a) \cdot_{\lambda} (a \vartriangleright_\lambda c)\\
&=& (a \circ_{\lambda} b)\cdot_{\lambda} a^{\lambda} \cdot_{\lambda} (a \circ_{\lambda} c), \quad \forall \lambda\in \Lambda, a,b,c\in G,
\end{eqnarray*}
which implies that $(\GGG,\{\cdot_\lambda\}_{\lambda\in \Lambda},\{\circ_\lambda\}_{\lambda\in \Lambda})$ is a dynamical skew brace.

Moreover, by Proposition \ref{homo-sub-d-group}, we can obviously deduce that $\Psi$ is a homomorphism of dynamical skew braces from $(\GGG,\{\cdot_\lambda\}_{\lambda\in \Lambda},\{\circ_\lambda\}_{\lambda\in \Lambda})$ to $(\HHH,\{\cdot'_\lambda\}_{\lambda\in \Lambda},\{\circ'_\lambda\}_{\lambda\in \Lambda})$.
\end{proof}

\begin{pro}\label{dsb-dpg}
Let $(\GGG,\{\cdot_\lambda\}_{\lambda\in \Lambda},\{\circ_\lambda\}_{\lambda\in \Lambda})$ be a dynamical skew brace. Define a multiplication $\vartriangleright_\lambda:G \times G \to G$ by
\vspace{-.2cm}
$$a \vartriangleright_\lambda b := a^{\lambda} \cdot_{\lambda} (a \circ_\lambda b),\quad\forall a,b\in G, $$
where $a^{\lambda}$ is the inverse of $a$ in $(G,\Lambda,\{\cdot_\lambda\}_{\lambda\in \Lambda})$. Then
$(\GGG, \{\cdot_\lambda\}_{\lambda\in \Lambda},\{\vartriangleright_\lambda\}_{\lambda\in \Lambda})$ is a dynamical post-group.

Let $\Psi:\GGG\to\HHH$ be a homomorphism of dynamical skew braces from $ (\GGG,\{\cdot_\lambda\}_{\lambda\in \Lambda},\{\circ_\lambda\}_{\lambda\in \Lambda})$ to $
(\HHH,  \{\cdot'_\lambda\}_{\lambda\in \Lambda},\{\circ'_\lambda\}_{\lambda\in \Lambda})$. Then $\Psi$ is a homomorphism of dynamical post-groups from $(\GGG,\{\cdot_\lambda\}_{\lambda\in \Lambda},\{\vartriangleright_\lambda\}_{\lambda\in \Lambda})$ to
$(\HHH,\{\cdot'_\lambda\}_{\lambda\in \Lambda},\{\vartriangleright'_\lambda\}_{\lambda\in \Lambda})$.
\end{pro}

\begin{proof}
Given any $\lambda\in \Lambda, a,c \in G$, there exists an element $b:=\bar{a}^{\lambda} \circ_{\phi_G(\lambda,a)} (a \cdot_{\lambda} c) $ such that
{\small
\begin{eqnarray*}
a \vartriangleright_\lambda b &=& a^{\lambda} \cdot_{\lambda} (a \circ_\lambda b)\\
&=& a^{\lambda} \cdot_{\lambda} \Big(a \circ_\lambda (\bar{a}^{\lambda} \circ_{\phi_G(\lambda,a)} (a \cdot_{\lambda} c)) \Big)\\
&=& a^{\lambda} \cdot_{\lambda} \Big( (a\circ_\lambda \bar{a}^{\lambda})\circ_\lambda (a \cdot_{\lambda} c) \Big)\\
&=& a^{\lambda} \cdot_{\lambda} (a \cdot_{\lambda} c)\\
&=& c,
\end{eqnarray*}
}
which implies $a\vartriangleright_\lambda -$ is surjective. Moreover, if $a\vartriangleright_\lambda b = a\vartriangleright_\lambda c,$ which is equivalent to $a \circ_\lambda b=a \circ_\lambda c$, then we have
{\small
\begin{eqnarray*}
b&=& (\bar{a}^{\lambda} \circ_{\phi_G(\lambda,a)} a) \circ_{\phi_G(\lambda,a)} b\\
&=& \bar{a}^{\lambda} \circ_{\phi_G(\lambda,a)} (a \circ_{\phi_G(\phi_G(\lambda,a),\bar{a}^{\lambda})} b)\\
&=& \bar{a}^{\lambda} \circ_{\phi_G(\lambda,a)} (a \circ_\lambda b)\\
&=& \bar{a}^{\lambda} \circ_{\phi_G(\lambda,a)} (a \circ_\lambda c)\\
&=& (\bar{a}^{\lambda} \circ_{\phi_G(\lambda,a)} a) \circ_{\phi_G(\lambda,a)} c\\
&=& c,
\end{eqnarray*}
which implies that $a\vartriangleright_\lambda -$ is injective. We also have
\begin{eqnarray*}
a\vartriangleright_\lambda (b \cdot_{\phi_{G}(\lambda,a)} c)
&=& a^{\lambda} \cdot_{\lambda} (a  \circ_{\lambda} (b \cdot_{\phi_{G}(\lambda,a)} c ))\\
&=&  a^{\lambda} \cdot_{\lambda} (a \circ_{\lambda} b) \cdot_{\lambda} a^{\lambda} \cdot_{\lambda} (a \circ_{\lambda} c)\\
&=&(a \vartriangleright_\lambda b)\cdot_\lambda (a \vartriangleright_\lambda c),
\end{eqnarray*}
}
and
{\small
\begin{eqnarray*}
(a \circ_{\lambda} b) \vartriangleright_\lambda c
&=& (a \circ_{\lambda} b)^{\lambda} \cdot_{\lambda} ((a \circ_{\lambda} b) \circ_{\lambda} c)\\
&\overset{\eqref{equ-d-brace}}{=}&  a^{\lambda} \cdot_{\lambda} (a \circ_{\lambda} b^{\phi_G(\lambda,a)}) \cdot_{\lambda} a^{\lambda} \cdot_{\lambda} ((a \circ_{\lambda} b) \circ_{\lambda} c)\\
&=&  a^{\lambda} \cdot_{\lambda} (a \circ_{\lambda} b^{\phi_G(\lambda,a)}) \cdot_{\lambda} a^{\lambda} \cdot_{\lambda} (a \circ_{\lambda} (b \circ_{\phi_G(\lambda,a)} c))\\
&=&  a^{\lambda} \cdot_{\lambda} \Big(a  \circ_{\lambda} (b^{\phi_G(\lambda,a)} \cdot_{\phi_G(\lambda,a)}  (b\circ_{\phi_G(\lambda,a)} c ) ) \Big)\\
&=&  a^{\lambda} \cdot_{\lambda} (a \circ_{\lambda} (b \vartriangleright_{\phi_G(\lambda,a)}c))\\
&=&  a \vartriangleright_{\lambda} (b \vartriangleright_{\phi_G(\lambda,a)}c),
\end{eqnarray*}
}
for all $\lambda\in \Lambda, a,b,c \in G$. Therefore,
$(\GGG, \{\cdot_\lambda\}_{\lambda\in \Lambda},\{\vartriangleright_\lambda\}_{\lambda\in \Lambda})$ is a dynamical post-group.

For the second conclusion, for all $\lambda\in \Lambda, a,b \in G$, we have
{\small
\begin{eqnarray*}
\Psi_{\lambda} (a \vartriangleright_\lambda b)
&=& \Psi_{\lambda} (a^{\lambda} \cdot_\lambda (a \circ_{\lambda} b))\\
&=& \Psi_{\lambda} (a^{\lambda}) \cdot'_{\lambda}  \Psi_{\lambda} (a \circ_{\lambda} b)\\
&\overset{\eqref{homo-d-skew-brace2}}{=}& \Psi_{\lambda} (a)^{\lambda} \cdot'_{\lambda} (\Psi_{\lambda}(a) \circ'_{\lambda} \Psi_{\phi_G(\lambda,a)}(b) )\\
&=& \Psi_{\lambda} (a) \vartriangleright'_\lambda   \Psi_{\phi_G(\lambda,a)}(b),
\end{eqnarray*}
}
which implies that $\Psi$ is a homomorphism of the corresponding dynamical post-groups from $(\GGG,\{\cdot_\lambda\}_{\lambda\in \Lambda},\{\vartriangleright_\lambda\}_{\lambda\in \Lambda})$ to $(\HHH,\{\cdot'_\lambda\}_{\lambda\in \Lambda},\{\vartriangleright'_\lambda\}_{\lambda\in \Lambda})$.
\end{proof}

\begin{thm}
Proposition \ref{dpg-dsb} gives a functor $\frkF:{\bf DPG} \to {\bf DSB}$ and Proposition \ref{dsb-dpg} gives a functor $\frkG:{\bf DSB} \to {\bf DPG}$. Moreover, they give an isomorphism between the categories ${\bf DPG}$ and  ${\bf DSB}$.
\end{thm}

\begin{proof}
Let $(\GGG,\{\cdot_\lambda\}_{\lambda\in \Lambda},\{\vartriangleright_\lambda\}_{\lambda\in \Lambda})$ and $(\HHH,\{\cdot'_\lambda\}_{\lambda\in \Lambda},\{\vartriangleright'_\lambda\}_{\lambda\in \Lambda})$ be two dynamical post-groups, and $\Psi$ be a homomorphism  from $(\GGG,\{\cdot_\lambda\}_{\lambda\in \Lambda},\{\vartriangleright_\lambda\}_{\lambda\in \Lambda})$ to $(\HHH,\{\cdot'_\lambda\}_{\lambda\in \Lambda},\{\vartriangleright'_\lambda\}_{\lambda\in \Lambda})$. By Proposition \ref{dpg-dsb} and \ref{dsb-dpg}, we have $\frkG\frkF(\GGG,\{\cdot_\lambda\}_{\lambda\in \Lambda},\{\vartriangleright_\lambda\}_{\lambda\in \Lambda})=(\GGG,\{\cdot_\lambda\}_{\lambda\in \Lambda},\{\vartriangleright_\lambda\}_{\lambda\in \Lambda})$ and $\frkG\frkF(\Psi)=\Psi$. Thus, $\frkG\frkF=\Id$. Similarly, we can check that $\frkF\frkG=\Id$. Therefore, the functor $\frkF$ is the inverse of the functor $\frkG$. The proof is completed.
\end{proof}

\begin{ex}{\rm
Let $(\mathbb{R},+,\cdot)$ be the field of real numbers. Then $(\mathbb{R},\Lambda=\mathbb{R},\phi,\{\cdot_{\lambda}\}_{\lambda\in \mathbb{R}},\{\circ_{\lambda}\}_{\lambda\in \mathbb{R}})$ is a dynamical skew brace, where $\cdot_{\lambda}:=+$ for all $\lambda\in \mathbb{R}$ and $\circ_{\lambda}:\mathbb{R} \times \mathbb{R} \to \mathbb{R}$, $\phi:\mathbb{R} \times \mathbb{R}\to \mathbb{R}$ are given in Example \ref{ex-d-group-R}.
In fact, this dynamical skew brace is equivalent to the dynamical post-group given in  Example \ref{d-post-group-R}.
}\end{ex}

\begin{ex}{\rm
Let $G=(\mathbb{Z}_3,+)$ be a cyclic group of order $3$ and let $\Lambda=\{\lambda_1,\lambda_2,\lambda_3\}$. Then $(G=\mathbb{Z}_3, \Lambda=\{\lambda_1,\lambda_2,\lambda_3\},\phi,\{\cdot_{\lambda_i}:=+ \}_{\lambda_i \in \Lambda},\{\circ_{\lambda_i}\}_{\lambda_i\in \Lambda})$ is a dynamical skew brace, where $\circ_{\lambda_i}$ is given in Example \ref{ex-d-group}. Actually, this dynamical skew brace is equivalent to the dynamical post-group given in Example \ref{ex-d-post-group}.
}\end{ex}

By Theorem \ref{DPG-BDG} and Proposition \ref{dsb-dpg}, we have the following conclusion.

\begin{cor}\label{d-skew-brace-solution}
Let $(\GGG,\{\cdot_\lambda\}_{\lambda\in \Lambda},\{\circ_\lambda\}_{\lambda\in \Lambda})$ be a dynamical skew brace. Define $R_G:\Lambda \times G \times G \to G \times G$ by
\begin{eqnarray*}
  R_G(\lambda,a,b)=(a^\lambda \cdot_{\lambda} (a \circ_\lambda b),\overline{a^\lambda \cdot_{\lambda} (a \circ_\lambda b)}^{\lambda} \circ_{\phi_G(\lambda,a^\lambda \cdot_{\lambda} (a \circ_\lambda b))} (a \circ_{\lambda} b) ),\quad \forall \lambda \in \Lambda,~a,b\in G.
\end{eqnarray*}
Then $R_G$ is a non-degenerate solution of the dynamical Yang-Baxter equation on the dynamical set $\GGG=(G,\Lambda,\phi_G)$.
\end{cor}

\noindent
{\bf Acknowledgements.}  This research is supported by NSFC (12271265, 12261131498, 12471060, W2412041) and
the Fundamental Research Funds for the Central Universities and Nankai Zhide Foundation.

\smallskip
\noindent
{\bf Declaration of interests. } The authors have no conflicts of interest to disclose.

\smallskip
\noindent
{\bf Data availability. } Data sharing is not applicable as no data were created or analyzed.

\end{document}